\documentclass[10pt,twocolumn,letterpaper]{article}

\usepackage{mpi}
\usepackage[utf8]{inputenc}
\usepackage{graphicx}
\usepackage{epsfig}
\usepackage{color}
\usepackage{times}
\usepackage{amsmath,mathtools}
\usepackage{amsthm}	
\usepackage{amssymb}
\usepackage[htt]{hyphenat}
\usepackage{enumerate}
\usepackage{subfigure}
\usepackage[algo2e,inoutnumbered,linesnumbered,algoruled,vlined]{algorithm2e}
\usepackage{algpseudocode}
\usepackage{booktabs}
\usepackage{tabularx}
\usepackage{array}
\usepackage{bm}
\usepackage{url}
\usepackage{bbm}
\usepackage{enumitem}
\usepackage[dvipsnames]{xcolor}
\usepackage{stackengine}
\usepackage{wrapfig}
\usepackage{multirow}
\usepackage{multicol}
\usepackage{tablefootnote}
\usepackage{comment}

\stackMath

\usepackage[pagebackref=true,breaklinks=true,letterpaper=true,colorlinks,bookmarks=false]{hyperref}

\usepackage{tabularx}
\newcommand{\PBS}[1]{\let\temp=\\#1\let\\=\temp}
\newcommand{\RBS}{\let\\=\tabularnewline}

\DeclareMathSymbol{@}{\mathord}{letters}{"3B}

\newcommand\blfootnote[1]{%
  \begingroup
  \renewcommand\thefootnote{}\footnote{#1}%
  \addtocounter{footnote}{-1}%
  \endgroup
}

\mpifinalcopy 
 
\newcommand{\x}{\mathbf{x}}

\newcommand{\Ag}{\mathbf{A}}

\newcommand{\bg}{\mathbf{b}}

\newcommand{\Id}{\mathbf{I}}

\newcommand{\Pm}{\mathbf{P}}

\newcommand{\one}{\mathbf{1}}

\newcommand{\X}{\mathbf{X}}

\newcommand{\q}{\mathbf{q}}
\newcommand{\Q}{\mathbf{Q}}

\newcommand{\vecm}{\mathrm{vec}}

\newcommand{\s}{\mathbf{s}}

\newcommand{\R}{\mathbb{R}}

\newcommand{\Pset}{\mathcal{P}}

\newcommand{\Graph}{\mathcal{G}}
\newcommand{\Vertex}{\mathcal{V}}
\newcommand{\Edge}{\mathcal{E}}
\newcommand{\Dom}{\mathcal{D}}
\newcommand{\Cycles}{\mathcal{C}}
\newcommand{\Paths}{\mathcal{P}}

\def\bfx{{\x}} 
\def\bfQ{{\Q}} 
\def\bfB{{{ \mathcal{B}}} }
\def\bfs{{\s}} 

\newtheorem{thm}{Theorem}
\newtheorem{remark}{Remark}

\newtheorem{prop}{Proposition}
\newtheorem{dfn}{Definition}

\DeclareMathOperator*{\argmin}{arg\min}
\DeclareMathOperator*{\argmax}{arg\max}

\usepackage{cleveref}
\usepackage{braket}

\newcommand{\parahead}[1]{\vspace{2mm}\noindent\textbf{#1}.\ }

\Crefname{assumption}{\textbf{H}\hspace{-3pt}}{\textbf{H}\hspace{-3pt}}
\crefname{assumption}{\textbf{H}}{\textbf{H}}

\Crefname{assumption}{\textbf{H}\hspace{-3pt}}{\textbf{H}\hspace{-3pt}}
\crefname{algorithm}{\text{Alg.}}{\text{Alg.}}
\crefname{assumption}{\textbf{H}}{\textbf{H}}
\crefname{equation}{\text{Eq}}{\text{Eq}}
\crefname{definition}{\text{Dfn.}}{\text{Dfn.}}
\crefname{lemma}{\text{Lemma}}{\text{Lemma}}
\crefname{dfn}{\text{Dfn.}}{\text{Dfn.}}
\crefname{thm}{\text{Thm.}}{\text{Thm.}}
\crefname{tab}{\text{Tab.}}{\text{Tab.}}
\crefname{fig}{\text{Fig.}}{\text{Fig.}}
\crefname{table}{\text{Tab.}}{\text{Tab.}}
\crefname{figure}{\text{Fig.}}{\text{Fig.}}
\crefname{section}{\text{Sec.}}{\text{Sec.}}
\crefname{prop}{\text{Prop.}}{\text{Prop.}}

\renewcommand{\etal}{\textit{et al}. }
\renewcommand{\ie}{\textit{i}.\textit{e}. }

\newcommand{\insertimageC}[5]{ 
\begin{figure}[#5]
\centering
\includegraphics[width=#1\linewidth, clip=true]{figures/#2}
\caption{#3}
\label{#4}
\end{figure}
}

\newcommand{\insertimageStar}[5]{ 
\begin{figure*}[#5]
\centering
\includegraphics[width=#1\linewidth, clip=true]{figures/#2}
\caption{#3}
\label{#4}
\end{figure*}
}

\begin{document}

\title{Quantum Permutation Synchronization} 

\author{Tolga Birdal$^{1, \boldsymbol{\star}}$ \hspace{1.8em} Vladislav Golyanik$^{2, \boldsymbol{\star}}$ \hspace{1.8em} Christian Theobalt$^2$ \hspace{1.8em} Leonidas Guibas$^1$ \vspace{10pt}\\
$^1$Stanford University 
\textbf{\hspace{3em}} 
$^2$Max Planck Institute for Informatics, SIC 
} 

\maketitle

\blfootnote{$\boldsymbol{^\star}$ both authors contributed equally to this work} 

\begin{abstract} 
We present \textbf{QuantumSync}, the first quantum algorithm for solving a synchronization problem in the context of computer vision. In particular, we focus on permutation synchronization which involves solving a non-convex optimization problem in discrete variables. We start by formulating synchronization into a quadratic unconstrained binary optimization problem (QUBO). While such formulation respects the binary nature of the problem, ensuring that the result is a set of permutations requires extra care. Hence, we: (i) show how to insert permutation constraints into a QUBO problem and (ii) solve the constrained QUBO problem on the current generation of the adiabatic quantum computers D-Wave. Thanks to the quantum annealing, we guarantee global optimality with high probability while sampling the energy landscape to yield confidence estimates. Our proof-of-concepts realization on the adiabatic D-Wave computer demonstrates that quantum machines offer a promising way to solve the prevalent yet difficult synchronization problems. 
\end{abstract} 

\vspace{-20pt} 
\vspace{-2mm}\section{Introduction}\label{sec:intro}

Computer vision literature accommodates a myriad of efficient and effective methods for processing rich information from widely available 2D and 3D cameras. 
Many algorithms at our disposal excel at processing single frames or sequence of images such as videos~\cite{rempe2020caspr,deng2020deep,birdal2019generic}. 
Oftentimes, a scene can be observed by multiple cameras, from different, usually unknown viewpoints. 
To this date it remains an open question how to \emph{consistently} combine the multi-view cues acquired without respecting any particular order \cite{Golyanik_MBGA2020}. 

\begin{figure}[t]
    \centering
    \includegraphics[width=\linewidth]{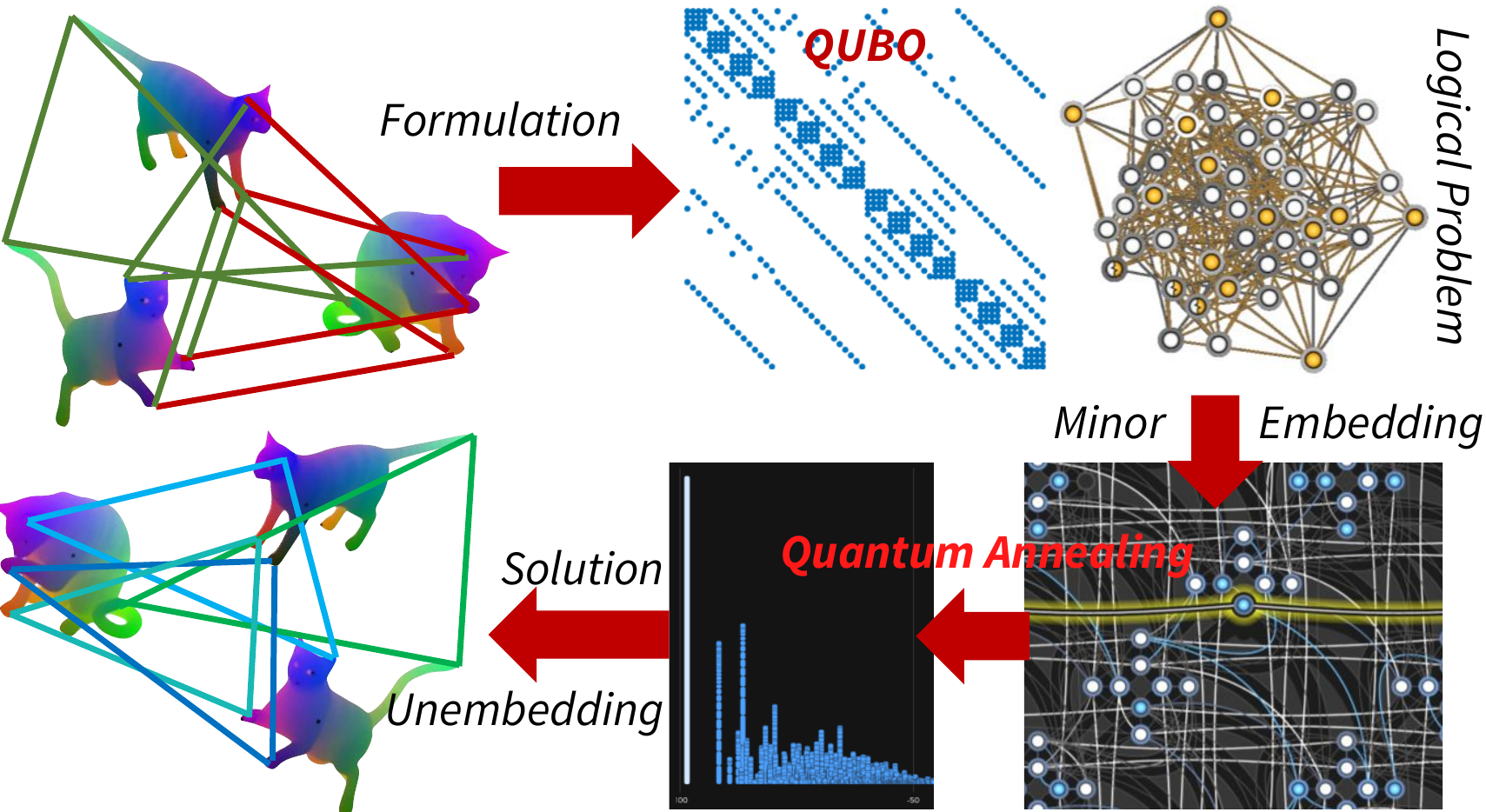} 
    \caption{{Overview of \textbf{QuantumSync}.} 
    QuantumSync formulates permutation synchronization as a QUBO and embeds its logical instance on a quantum computer. After running multiple anneals, it selects the lowest energy solution as the global optimum. 
    \vspace{-5mm}} 
    \label{fig:teaser} 
\end{figure} 

\emph{Synchronization}~\cite{simons1990overview,giridhar2006distributed,singer2011angular} is one of the proposed solutions to the aforementioned problem. 
On an abstract note, it involves distributing the discrepancies over the graph connecting multiple viewpoints such that the estimates are consistent across all considered nodes. 
To this end, synchronization simultaneously \emph{averages} the pairwise local information into a global  one~\cite{govindu2014averaging,hartley2013rotation,Arrigoni2019,birdal2020synchronizing}. 
This procedure is a fundamental piece of most state-of-the-art multi-view reconstruction and multi-shape analysis  pipelines~\cite{salas2013slam++,cadena2016past,carlone2015initialization} because it heavy-lifts the global constraint satisfaction while respecting the geometry of the parameters. 
In fact, most of the multiview-consistent inference problems can be expressed as some form of a synchronization~\cite{zhang2019path,birdal2017cad}. 

In this paper, our focus is \emph{permutation synchronization}, where the edges of the graph are labeled by permutation matrices denoting the correspondences either between two 2D images or two 3D shapes. 
Specifically, we seek to find an absolute ordering for each point of each frame which sorts all the corresponding points into the same bin. Unfortunately, this problem by definition involves a combinatorial non-convex optimization, for which attaining the global minimum is intractable under standard formulations targeting classical von Neumann computers.
This difficulty have encouraged scholars to seek \emph{continuous relaxations} for which either a good local optimum~\cite{birdalSimsekli2018,birdal2019probabilistic} or a closed form solution~\cite{arrigoni2016spectral,huang2013consistent,arrigoni2017synchronization} could be found. 
The solution to this approximately equivalent relaxed problem should then be rounded back to the permutation matrices to report a valid result. 
Ideally, we would like to avoid such approximations and work with the original set of constraints at hand. This is exactly what we propose to do in this work. 
Because on a classical computer we cannot speak of a global optimality guarantee of our discrete problem, we turn our attention to a new family of processors, \textit{i.e.,} \emph{quantum computers}.

A quantum computer is a computing machine which takes advantage of quantum effects  such as \emph{quantum superposition}, \emph{entanglement}, \emph{tunnelling} and \emph{contextuality} to solve problems notoriously difficult  (\textit{i.e.,}~$\mathcal{NP}$-hard) on a classical computer. 
Numerous quantum algorithms have demonstrated improved computational complexity finally reaching the desired \emph{supremacy}~\cite{arute2019quantum}. 
Currently, we see the first practical uses of quantum computers thanks to the programmable quantum processing units (QPU) available to the research community such as D-Wave \cite{DWAVE} or IBM~\cite{IBM}. 
As this new generation of computers offers a completely different computing paradigm, directly porting classical problem formulations is far from being trivial. In fact, oftentimes we are required to revise the problem at hand altogether~\cite{Grover96afast}. 

This paper shows how to formulate the classical permutation synchronization in terms of a quadratic unconstrained binary optimization (QUBO) that is quantum computer friendly. QUBO optimizes for \emph{binary} variables and not permutations, and, hence, we are required to ensure that permutation constraints are respected during optimization. To this end, we turn all constraints into linear ones and incorporate them into QUBO. Finally, we embed our problem to a real quantum computer and show that it is highly likely to achieve the lowest energy solution, \emph{the global optimum} (see \cref{fig:teaser}). 
Our \textbf{contributions} are as follows: 
\begin{enumerate}[noitemsep,topsep=1pt]
    \item The first, to the best of our knowledge, formulation of the classical synchronization 
    in a form consumable by an adiabatic quantum computer (AQC); 
    \item We show how to introduce permutations as linear constraints into the QUBO problem; 
    \item We numerically verify the validity of our formulation in simulated experiments as well as on a real AQC (for the first time, on D-Wave Advantage $1.1$);  
    \item We perform extensive ablation studies giving insights into this new way of computing. 
\end{enumerate}
We obtain highly probable global optima up to either of: (i) eight views and three points per view, (ii) seven views with four points, or (iii) five points with three views. 
We experiment for the first time on an AQC with $5k$ qubits and perform extensive evaluations and comparisons to classical methods. 
Our approach can also be classified as the first method for quantum matching of multiple point sets. 
While we are the first to implement synchronization on quantum hardware, our evaluations and tests are proof of concepts as truly practical quantum computing is still a leap away. 

\section{Related Work}\label{sec:related} 

\parahead{Synchronization} 
The art of consistently recovering absolute quantities from a collection of ratios, \emph{synchronization}, is now the de-facto choice when it comes to bringing functions on image/shape collections into unison~\cite{salas2013slam++,cadena2016past,carlone2015initialization}. The problem is now very well studied, enjoying a rich set of algorithms. Primarily, there exists a plethora of works on group structures arising in different applications~\cite{govindu2014averaging,govindu2004lie,birdal2019probabilistic,arrigoni2017synchronization,Arrigoni2019,huang2019tensor,hartley2013rotation,Arrigoni2019,wang2013exact,chaudhury2015global,thunberg2017distributed,tron2014distributed,arrigoni2016spectral,bernard2015solution}. Some of the proposed solutions are closed form~\cite{arrigoni2016spectral,arrigoni2017synchronization,maset2017,Arrigoni2019} or minimize robust losses~\cite{hartley2011l1,chatterjee2017robust}. Others address certifiability~\cite{rosen2019se} and global optimality~\cite{briales2017cartan}. Bayesian treatment or uncertainty quantification is also considered~\cite{tron2014statistical,birdalSimsekli2018,birdal2020synchronizing,birdal2019probabilistic} as well as low rank matrix factorizations~\cite{bernard2018}. Recent algorithms tend to incorporate synchronization into deep learning~\cite{huang2019learning,gojcic2020learning}.  
This work concerns with synchronizing correspondence sets, otherwise known as \emph{permutation synchronization} (PS)~\cite{pachauri2013solving}. This sub-field also attracted a descent amount of attention: low-rank formulations~\cite{yu2016globally,wang2018multi}, convex programming~\cite{hu2018distributable}, multi-graph matching\cite{schiavinato2017synchronization}, distributed optimization~\cite{hu2018distributable} or Riemannian optimization~\cite{birdal2019probabilistic}. 

All of these approaches try to cope with the intrinsic non-convexity of the synchronization problem one way or another. Unfortunately, solving our problem on a classical computer is notoriously difficult. To the best of our knowledge, we are the firsts to address this problem through the lens of a new paradigm, \emph{adiabatic quantum computing}. 
\noindent\textbf{Quantum computing.} 
Since its motivation in the 1980s \cite{Manin1980, Feynman1982}, quantum computing  has become an active research area, both from the hardware \cite{Chuang1998, Kane1998, DiVincenzo2000, Wesenberg2009, Morello2010, Zwanenburg2013, Lanting2014etal, Lekitsch2017, Lanting2017} and algorithmic side \cite{Shor1994, Simon1994, BonehLipton1995, Grover96afast, Farhi2001, Hallgren2005, vanDamShparlinski2008, Harrow2009, Clader2013, Aimeur2013, Lloyd2014, Schuld2016, Neukart2017, DonisVelaGarciaEscartin2018, Zhao2019, Suksmono2019}. 
Quantum methods offering speedup compared to the classical counterparts have been demonstrated for
domains such as applied number theory \cite{Shor1994, BonehLipton1995, vanDamShparlinski2008,  DonisVelaGarciaEscartin2018}, 
linear algebra \cite{Hallgren2005, Harrow2009, LeGall2012, TaShma2013,  Suksmono2019}, machine learning \cite{Aimeur2013, Lloyd2014, Zhao2019} and simulation of physical systems \cite{Zalka1998, Ortiz2001, Berry2007, Kassal2008, Jordan2012}, among others. 

Quantum annealing (QA)~\cite{KadowakiNishimori1998, Brooke1999} and the adiabatic quantum evolution algorithm by Farhi \textit{et  al.}~\cite{Farhi2001} have triggered the development of adiabatic quantum computers (AQC). 
In the last decade, the technology has matured and became  accessible remotely for test and research purposes 
with the help of D-Wave 
\cite{DWave_Leap}. 
A recent benchmarking of D-Wave AQC \cite{Denchev2016} has shown that for energy landscapes with large and tall barriers, quantum annealing can achieve speed-ups of up to eight orders of magnitude compared to simulated annealing \cite{Kirkpatrick1983} running on a single core. 
AQC has also been successfully applied to traffic flow optimization while outperforming  classical methods \cite{Neukart2017}. 

\noindent\textbf{Quantum computer vision.} 
Several quantum methods for computer vision problems have been proposed in the literature including algorithms for image recognition \cite{Neven2012}, classification \cite{Boyda2017, nguyen2020regression} and facial feature learning by low-rank matrix factorization \cite{OMalley2018}. 
Recently, D-Wave has been applied to redundant object removal in object detection  \cite{LiGhosh2020}. 
While a QUBO formulation of non-maximum suppression 
was already known in the literature \cite{RujikietgumjornCollins2013}, it has been  improved in \cite{LiGhosh2020} for D-Wave 2X 
with solutions outperforming several classical  methods. 
Quantum approach for point set alignment \cite{pointCorrespondence,pointCorrespondenceSupp} solves a related problem to ours. 
It approximates rotation matrices in the affine space by basis elements which can be summed up according to the measured bitstring. 
However, this formulation cannot be easily extended to the multi-view case and does not support permutation constraints. 
In contrast, our method solves a multi-way matching problem, and we successfully deploy it on a real AQC Advantage $1.1$. 
Our use of assignment constraints is also new. 
Concurrently to us, a similar form of permutation matrix  constraints was proposed by  Benkner~\etal~\cite{SeelbachBenkner2020} for matching two graphs on AQC. 
Several previous works formulate similar penalties in different ways. 
Stollenwerk~\etal\cite{stollenwerk2019flight} address the flight assignment problem on AQC. 
Formulation of the graph isomorphism problem can include permutations~\cite{zick2015experimental, gaitan2014graph}. 
In~\cite{zick2015experimental}, an individual variable for every vertex pair of the same degree in a graph is allocated, which subsequently leads to a QUBO. 
Permutations can also be converted to a table with binary entries added as a penalty term to the target Hamiltonian~\cite{gaitan2014graph}. 

\section{Preliminaries and Technical Background}
\subsection{Synchronization}\label{sec:bg:sync}

We model the multi-view configuration as a connected undirected graph $\Graph=(\Vertex=\{1,2,\dots,n\}, \Edge\subset [n] \times [n])$ where $|\Edge|=m$ and if $(i,j)\in\Edge$ then $(j,i)\in\Edge$. Each vertex $v_i\in \Vertex$ (\textit{e.g.,} image) is associated a domain $\Dom_i$ (\textit{e.g.,} ordered points). Each edge  $(i,j)\in\Edge$ is labeled with a function  $f_{i,j}\,:\,\Dom_i\mapsto\Dom_j$ (\textit{e.g.,} correspondence). We  will refer to the edge-related entities as \emph{relative} and node- (vertex-) related entities as \emph{absolute}. Thus, we aptly call $\{f_{i,j}\}_{i,j}$ as \emph{relative maps}. In this section, we  define and explain the necessary notions  following~\cref{fig:consistency}. 
{The proofs of the theorems given in this section can be found in our supplementary material.} 

\begin{dfn}[Path, Cycle and Null Cycle]
We define a \textbf{path} to be the ordered, unique index sequence $p=\{(i_1,i_2), (i_2,i_3), \cdots, (i_{n-1},i_n)\} \in \Paths$ along $\Graph$ connecting $v_{i_1}$ to $v_{i_n}$. It is called a \textbf{cycle}, if additionally the path traces back to the starting node. Finally, following~\cite{Arrigoni2019}, we denote a cycle $c\in\Cycles$ to be a \textbf{null-cycle} of $\Graph$ if the composition of functions along $c$ leads to the identity transformation:
\begin{equation}
\label{eq:nullcycle}
    f_c = f_{1,2}\circ f_{2,3}\cdots\circ f_{(n-1),n} \circ f_{n,1} = f_{\varnothing},
\end{equation}
where $f_{\varnothing}(\x) = \x,\, \forall \x\in\Dom_1$ is the identity map and $f_c$ denotes the composite function. Intuitively, $c$ is a non-empty path in which the only repeated vertices are the first and last vertices. Note that the  direction of the action of $\circ$ matters  (\textit{e.g.,} $f_{1,2}\circ f_{2,3} \neq  f_{2,3}\circ f_{1,2}$)  and is up to the  convention. 
\end{dfn}

\begin{dfn}[$k$-cycle]\label{dfn:kcycle}
We refer to a function mapping a vertex to itself as the \textbf{1-cycle}: $f_{i,i}=f_{\varnothing}$. Similarly, a {2-cycle} would be $f_c=f_{i,j}\circ f_{j,i}$ and so on~\cite{nguyen2011optimization,huang2013consistent}. 
\end{dfn}
\begin{dfn}[Cycle Consistency]
\label{dfn:cc}
We call the graph $\Graph$ to be cycle-consistent on $\Cycles$ if $f_c = f_{\varnothing} \,\, \forall c\in \Cycles$, 
where $\Cycles$ is the set of all cycles~\cite{guibas2019condition}.
\end{dfn}
The notion of cycle consistency for \emph{directed} graphs is known as \emph{path invariance}~\cite{zhang2019path} and differs from cycle consistency as shown in~\cref{fig:consistency}. In this paper, we further assume the maps belong to the \emph{general linear group} and are \emph{isomorphisms}, \textit{i.e.,}  $f_{ji}=f_{ij}^{-1}$. 

\begin{figure}[t]
    \centering
    \includegraphics[width=\columnwidth]{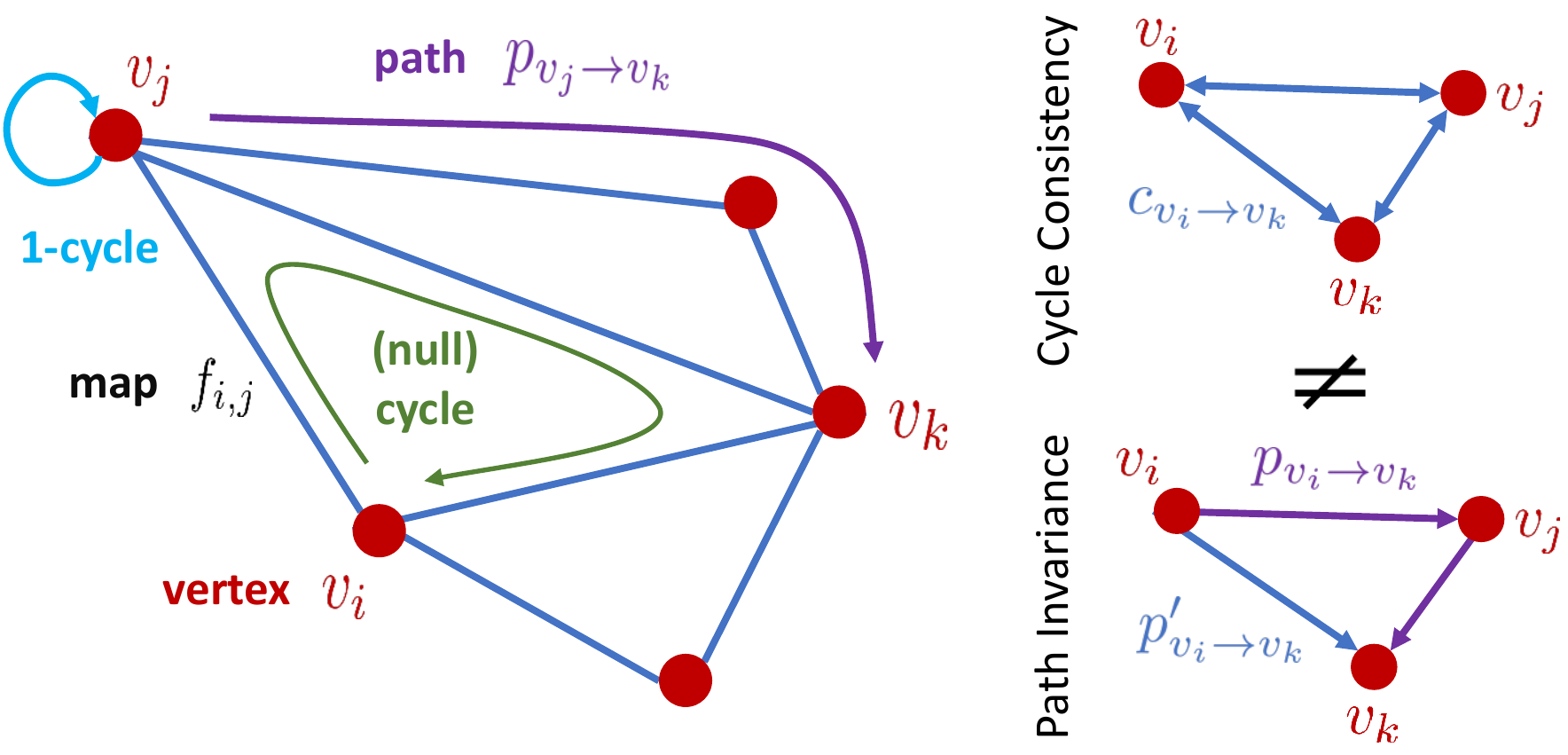}
    \caption{\textbf{(left)} Notation and illustration of the general setting we consider, \textbf{(right)} Cycle consistency \textit{vs} path invariance.\vspace{-3mm}}
    \label{fig:consistency}
\end{figure} 
\begin{remark}
Depending on the graph topology, the number of cycles may be exponential in the number of vertices. Hence, naively ensuring the consistency of large graphs according to~\cref{dfn:cc} quickly becomes intractable. Algorithms such as Guibas~\textit{et al.}~\cite{guibas2019condition} aim to satisfy the consistency of a subset of cycles $\bar{\Cycles}\subset\Cycles$ (bases), where enforcing consistency along these cycles induces consistency along all cycles of the input graph $\Graph$. However, as efficient selection of these \emph{cycle-consistency bases}~\cite{zach2010disambiguating} is a problem under investigation, we instead use the available group structure:
\end{remark}
\begin{thm}[Cycle Consistency by Construction]
\label{theorem:cycle_consistency} 
A \textbf{consistent vertex labeling} $\{f_i:\Vertex \mapsto \Dom_i\}_i$ where $(\Dom,\circ)$ forms a \textbf{group} with operation $\circ$, can be constructed by satisfying the following constraint for all $i$ and $j$:
\begin{align}
\label{eq:cc}
f_{i,j} = f_i \circ f_j^{-1}.
\end{align}
\end{thm}
\noindent Hence,~\cref{eq:cc} is called the \emph{cycle consistency constraint}.

\begin{dfn}[Synchronization]
\label{dfn:sync}
Synchronization is the procedure of finding a consistent labeling of $\Graph$ given a collection of ratios $\{f_{i,j}\}_{i,j}$ ensuring the cycle consistency of $\Graph$:
\begin{equation}\label{eq:sync}
    \argmin_{\{f_k\}_k}\sum\limits_{(i,j)\in\Edge}   d(\hat{f}_{i,j}, f_{i,j}).
\end{equation}
$\hat{f}_{i,j}=f_i\circ f_j^{-1}$ are the estimated ratios and $d(\cdot)$ is a group-specific distance metric.
\end{dfn}
\begin{remark}
A closer look to the problem reveals that it is \textbf{non-convex} due to the composition, but \textbf{convex when} the vertices $f_j$ are fixed during optimization of $f_i$. In fact, if $f_i$ is considered to be fixed, this problem resembles an averaging under the metric $d(\cdot)$. As for different $i$ we have different averages to compute, synchronization is often referred as \textbf{multiple averaging}~\cite{fredriksson2012simultaneous,hartley2013rotation,chatterjee2013efficient,hartley2011l1}.
\end{remark}
\begin{thm}[Gauge Freedom]\label{theorem:gauge_freedom} 
The problem in~\cref{dfn:sync} is subject to a freedom in the choice of the reference or the \emph{gauge}~\cite{belov2019geometry,chatterjee2013efficient}. In other words, the solution set to~\cref{eq:sync} can be transformed arbitrarily by a common $f_g$ \ie $f_i\gets f_i\circ f_g$ while still satisfying the consistency constraint.
\end{thm}
In practice, a gauge is fixed by setting one of the vertex labels to identity: $f_1 = f_{\varnothing}$.
\begin{dfn}[Permutation Matrix]
A \textit{permutation matrix} is defined as a sparse, square binary matrix, where each column or row contains only a single non-zero entry:
\begin{equation}
\label{eq:def:perm}
\mathcal{P}_n := \{\Pm \in \{0,1\}^{n\times n} : \Pm \one_n = \one_n\,,\,\one_n^\top \Pm = \one_n^\top\}.
\end{equation}
where $\one_n$ denotes a $n$-dimensional ones vector. Every $\Pm \in \mathcal{P}_n$ is a \textit{total} permutation matrix and $P_{ij}=1$ implies that point $i$ is mapped to element $j$. Note, $\Pm^\top=\Pm^{-1}$. We also denote the \textbf{product manifold} of $m$ permutations as $\mathcal{P}_n^m$. 
\end{dfn}
\begin{dfn}[Relative Permutation]
We define a permutation matrix to be a \textbf{relative map} if it is the ratio (or difference) of two group elements $(i\rightarrow j)$: $\Pm_{ij}=\Pm_i\Pm_j^\top$.
\end{dfn}
\begin{dfn}[Permutation Synchronization Problem]
For a redundant set of measures of ratios $\{\Pm_{ij}\}$, the permutation synchronization~\cite{pachauri2013solving} seeks to recover $\{\Pm_i\}$ for $i=1,\dots,N$ such that~\cref{eq:cc} is satisfied: $\Pm_{ij}=\Pm_i\Pm_j^{-1}$.
\end{dfn}
If the input data is noise-corrupted, this \textit{consistency} will not hold and to recover the \textit{absolute permutations} $\{\Pm_i\}$, some form of a \textit{consistency error} is minimized. 

\subsection{Adiabatic Quantum Computation}\label{ssec:quantum_permutation} 
Contemporary AQC can solve QUBO over a set of pseudo-Boolean functions of the following form: \vspace{-4pt} 
\begin{equation}\label{eq:QUBO_main} 
\argmin_{\bfx \in \bfB^n} \bfx^\top \bfQ \bfx + \bfs^\top \bfx, 
\vspace{-4pt} 
\end{equation} 
where $\bfB^n$ denotes the set of binary vectors of length $n$, $\bfQ \in \mathbb{R}^{n \times n}$ is a real symmetric matrix and $\bfs$ is a real $n$-dimensional vector. 
\eqref{eq:QUBO_main} is a frequent problem 
that is known to be $\mathcal{NP}$-hard on a classical computer. 
AQA operates with \textit{qubits} obeying the laws of quantum mechanics. 
In contrast to a classical bit, a qubit $\ket{\phi}$ can continuously transition between the states $\ket{0}$ and $\ket{1}$ (the equivalents of classical states $0$ and $1$) fulfilling the equation $\ket{\phi} = \alpha \ket{0} + \beta \ket{1}$, with probability amplitudes satisfying $|\alpha|^2 + |\beta|^2 = 1$. 
AQA algorithms for problems in a non-QUBO form first have to  
\textit{encode the problem} as QUBO  \eqref{eq:QUBO_main}, which defines the former in terms of \textit{logical} qubits 
and weights ($\bfs$ and $\bfQ$) between them. 
The logical problem is then \textit{minor-embedded} to the AQA's physical qubits graph, with methods such as \cite{Cai2014}. 
AQA interprets $\bfs$ and $\bfQ$ as qubit biases and couplings, respectively, and converts them to local magnetic fields on a QPU imposed on the qubits during \textit{anneallings}, \textit{i.e.,} free evolutions of the quantum-mechanical computing system. 
The search of the optimal $\bfx$ is performed by optimizing over the \textit{hidden}\footnote{in  the sense that $\alpha$, $\beta$ cannot be revealed (measurement postulate) and non-destructive copying of $\ket{\phi}$ is not possible (no-cloning theorem \cite{WoottersZurek1982})} 
states of $n$ qubits and consequently \textit{unembedding} and \textit{measuring} them as classical bitstrings. 
The latter are then passed to the \textit{solution interpretation} step, which decodes them in the context of the original problem. 

Before annealing, all $n$ qubits are initialized in ground states of an initial \textit{Hamiltonian} $\mathcal{H}_I$ that is easy to achieve as a superposition with equal probabilities of measuring $\ket{0}$ or $\ket{1}$ for every qubit. 
A Hamiltonian is an operator defining the energy spectrum of the system and--- interpreted for computational problems--- the space of all possible solutions. 
AQA performs a series of annealings, during which $\mathcal{H}_I$ continuously alters towards the problem Hamiltonian $\mathcal{H}_P$ under the influence of the local magnetic fields. 
This \emph{instantaneous} Hamiltonian $\mathcal{H}$ can be expressed as: 
\begin{equation}
    \mathcal{H}(\tau) = [1 - \tau]\,\mathcal{H}_I + \tau\,\mathcal{H}_P, 
\end{equation}  
with the local time variable $\tau \in [0; 1]$ transitioning from $0$ to $1$ during the annealing time $t$ (\textit{e.g.,} $20$ $\mu$s). 
According to the adiabatic theorem of quantum mechanics  \cite{BornFock1928}, the system will likely remain in the ground state of 
$\mathcal{H}_P$ by the end of the anneal, assuming a \textit{sufficiently long} $t$ and despite a highly non-convex energy landscape of the problem. 
This is due to the quantum effects of superposition, entanglement and  tunnelling. 
\textit{Superposition} enables the optimization to be performed on all possible qubit states simultaneously; it allows to operate on a $2^n$-dimensional state space spanned by $n$ qubits. 
During quantum computation, \textit{entangled} states are created, \textit{i.e.,} the states of multiple qubits which cannot be described independently from each other. 
\textit{Tunnelling} enables transition through the barriers (in geometric interpretation). 
For a more comprehensive overview of the AQC foundations, see \cite{KadowakiNishimori1998, Farhi2001, DWAVE} as well as Secs.~2 and 3 of \cite{pointCorrespondence}. 

\section{Quantum Synchronization}
\label{sec:method}
Suppose a multi-view configuration where we are given $n$ points in each of the $m$ views. Points in view $i$ relate to the ones in view $j$ via a permutation $\Pm_{ij}$ (see~\cref{sec:bg:sync}). $\Pm_{ij}$ is usually obtained independently for each pair in the edge set $\Edge$ and hence is \emph{noisy}. 
Following~\cite{birdal2019probabilistic}, we see the permutation synchronization as a probabilistic inference problem, where we will be interested in the following quantities: 
\begin{enumerate}[itemsep=2pt,topsep=1pt,leftmargin=*]
\item Maximum a-posteriori (MAP): 
\begin{align}\label{eq:map}
\X^\star = \argmax\limits_{\X \in \Pset_n^m} \log p(\X | \Pm )
\end{align}
where 
$\log p(\X | \Pm ) =^+ - \beta  \sum_{(i,j)\in \Edge} \|\Pm_{ij} - \X_i \X_j^\top \|^2_\mathrm{F} $, and $=^+$ denotes equality up to an additive constant. 
\item The full posterior distribution: $p(\X|\Pm) \propto p(\Pm,\X)$.
\vspace{1pt}
\end{enumerate}
Here, $\X^\star$ denotes the entirety of the sought permutations, and, similarly, $\Pm$ is the collection of all specified pairwise permutations.
The MAP estimate is often easier to obtain and useful in practice. On the other hand, samples from the full posterior can provide important additional information, such as \emph{uncertainty}. Not surprisingly, the latter is a much harder task, especially considering the discrete nature of our problem.
Each AQC algorithm includes problem encoding, 
minor embedding, 
sampling and solution interpretation, as described in~\cref{ssec:quantum_permutation}. 
In what follows, we will describe the problem  encoding and the preparation of  $\bfQ$  in~\cref{eq:QUBO_main}. 

\subsection{Permutation Synchronization as QUBO}
We first re-write the synchronization loss in~\cref{eq:map} as a quadratic assignment problem (QAP) that is more friendly for adiabatic optimization, and later insert the permutations as linear constraints into the formulation. 

\begin{prop}\label{prop1} 
Permutation synchronization under the Frobenius norm can be written in terms of a QUBO: 
\begin{align}
    \argmin\limits_{\{\X_i \in \Pset_n\}} \sum_{(i,j)\in \Edge} \|\Pm_{ij} - \X_i \X_j^\top \|^2_\mathrm{F}=\argmin\limits_{\{\X_i \in \Pset_n\}} \, \x^\top \Q^\prime \x.\nonumber
\end{align}
Here, $\x=[\cdots \x_i^\top \cdots]^\top$ and $\x_i=\vecm(\X_i)$
where $\vecm(\cdot)$ acts as a vectorizer. $\Q^\prime$ is then a matrix of the form:
\begin{align}\label{eq:Q_prime} 
\Q^\prime=
    -\begin{bmatrix}
    \Id \otimes \Pm_{11} & \Id \otimes \Pm_{12} & \cdots & \Id \otimes \Pm_{1m}\\
    \Id \otimes \Pm_{21} & \Id \otimes \Pm_{22} & \cdots & \Id \otimes \Pm_{2m}\\
    \vdots & \vdots & \ddots & \vdots\\
    \Id \otimes \Pm_{m1} & \Id \otimes \Pm_{m2} & \cdots & \Id \otimes \Pm_{mm}\\
    \end{bmatrix}.
\end{align}
\end{prop}
\begin{proof}\renewcommand{\qedsymbol}{}
The steps are intuitive to follow and an expanded proof is included in the supplementary material:
\begin{align}
\X^\star &= \argmin\limits_{\X\in \Pset_n^m} \sum_{(i,j)\in \Edge} \|\Pm_{ij} - \X_i \X_j^\top \|^2_\mathrm{F} \\
&= \argmin\limits_{\X\in \Pset_n^m} \,2N^2 n-2\sum_{(i,j)\in \Edge} \mathrm{tr}(\X_j\X_i^\top\Pm_{ij})\\
&=\argmin\limits_{\X\in \Pset_n^m} \,-\sum_{(i,j)\in \Edge} \vecm(\X_i)^\top (\Id \otimes \Pm_{ij} )\vecm(\X_j) \nonumber\\
&=\argmin\limits_{\X\in \Pset_n^m} \, \x^\top \Q^\prime \x.\label{eq:quprog}
\end{align}\vspace{-5mm}
\end{proof}
The Hessian of this problem is given by $\Q^\prime$ itself. Hence, the problem is only convex and solvable by algorithms such as \emph{interior-point} or \emph{trust-region} methods when $\Q^\prime$ is \emph{positive definite}. The indefinite problems can be solved via \emph{active-set} methods if the variables are relaxed to the set of reals. For discrete variables and when $\Q^\prime$ is not {positive definite}, this is a variant of $\mathcal{NP}$-hard \emph{integer quadratic problem}. We will instead show how to use the recent quantum computers to obtain the global minimum. 

\parahead{Formulating the QUBO synchronization}
The problem in~\cref{eq:quprog} has an added difficulty of being a discrete combinatorial optimization problem over the \emph{product manifold of permutations}. Replacing permutations with different choices of matrices lead to different relaxations, for example: (i) positivity allows for semi-definite programming~\cite{huang2013consistent}, (ii) orthonormality allows for spectral solutions~\cite{maset2017}, and (iii) doubly stochastic relaxation can allow for Riemannian optimization~\cite{birdal2019probabilistic}. However, all of these methods have to be followed by a projection step onto the discrete Permutohedron, often cast as an assignment problem and solved via the celebrated Hungarian algorithm. QUBO enables us to solve this problem \emph{without continuous relaxations} in a globally optimal manner by rephrasing~\cref{eq:quprog} in terms of binary variables $\bfB$ at our disposal:
\begin{align}
\argmin\limits_{\x \in \bfB} \, \x^\top \Q^\prime \x. 
\end{align}

\parahead{Permutations as linear constraints}
The binary variables $\x\in\bfB$ are a superset of the product-permutations. In other words, the solution to~\cref{eq:quprog} need not result in permutations once the matrices corresponding to each node are extracted. We propose to encourage the solution towards a set of permutations by introducing linear constraints $\Ag\x=\bg$ such that the optimization adheres to the definition of a permutation: rows and columns sum to one as in~\cref{eq:def:perm}. Given $\x_i=\vecm(\X_i)$, this amounts to having $\bg_i=\one$ and 
\begin{equation} 
    \Ag_i=\begin{bmatrix}
    \Id \otimes \one^\top \\
    \one^\top \otimes \Id
    \end{bmatrix}.
\end{equation}
Put simply, the matrix $\Ag_i$ is assembled as follows: in row $j$ with $1\leq j \leq n$, the ones are placed in columns $(j-1)\cdot n + 1 $ to $(j )\cdot n $. In a row $j$ with $ j>n$, ones will be placed at $ (j-n )+ p\cdot n $ for $p \in \{ 0,...,n-1 \} $. To enforce the permutation-ness of all the individual $\x_i$ that make up $\x\in\R^{n^2\times m}$, we construct a $n^2 \times 2n$ block-diagonal matrix $\Ag = \mathrm{diag}(\Ag_1, \Ag_2, \dots, \Ag_m)$. 

\parahead{Introducing linear constraints into QUBO}
We now extend our formulation by introducing the equality constraints $\Ag\x=\bg$ into the optimization.

\begin{prop}\label{prop2}
The constrained minimization: 
\begin{align}
\argmin\limits_{\x \in \bfB} \, \x^\top \Q^\prime \x \quad\text{ s.t. }\quad \Ag\x=\bg
\end{align}
can be turned into an (unconstrained) QUBO 
\begin{align}\label{eq:unconstrained_QUBO} 
    \argmin\limits_{\x \in \bfB} \, \x^\top \Q \x + \s^\top\x, 
\end{align}
where $\Q=\Q^\prime+\lambda\Ag^\top\Ag$ and $\s=-2\lambda\Ag^\top\bg$.
\end{prop}
Both $\Q^\prime$ and $\Q$ are sparse matrices. We provide their sparsity patterns and the proof of~\cref{prop2} in our supplement. 
We finally map this modified QUBO $(\Q,\s)$ onto D-Wave. 

\parahead{Encoding, Sampling and Interpretation} 
We fix the gauge by letting the first matrix to be identity: $\X_1=\Id$. Once \textit{QuantumSync} terminates, the measured bitstring can be directly interpreted as the solution permutations after reordering into $m - 1$ matrices of dimension $n\times n$. To explain the posterior landscape, we propose to sample from many low-energy states using the same quantum annealing. Further details are presented in the supplement. 

\begin{table}[t]
  \centering
  \caption{Evaluations on Willow Dataset.}
  \setlength{\tabcolsep}{3pt}
  \resizebox{\columnwidth}{!}
  {
    \begin{tabular}{lcccc|c}
          & Car   & Duck  & Motorbike & Winebottle & Average \\
\cmidrule{2-6}    Exhaustive & \textcolor{black}{\textbf{0.84  $\pm$  0.104}} & \textbf{0.91  $\pm$  0.115} & \textbf{0.82  $\pm$  0.10} & \textbf{0.95  $\pm$  0.096} & \textbf{0.88  $\pm$  0.104} \\
    EIG   & \textcolor{gray}{0.81  $\pm$  0.083} & \textcolor{gray}{0.86  $\pm$  0.102} & \textcolor{gray}{0.77  $\pm$  0.059} & \textcolor{gray}{0.87  $\pm$  0.107} & \textcolor{gray}{0.83  $\pm$  0.088} \\
    ALS   & \textcolor{black}{\textbf{0.84  $\pm$  0.095}} & 0.90  $\pm$ 0.102 & 0.81  $\pm$  0.078 & 0.94  $\pm$  0.092 & 0.87  $\pm$  0.092 \\
    LIFT  & \textcolor{black}{\textbf{0.84  $\pm$  0.102}} & 0.90  $\pm$  0.103 & 0.81  $\pm$  0.078 & 0.94  $\pm$  0.092 & 0.87  $\pm$  0.094 \\
    Birkhoff & \textbf{0.84  $\pm$  0.094} & 0.90  $\pm$  0.107 & 0.81  $\pm$  0.079 & 0.94  $\pm$  0.093 & 0.87  $\pm$  0.093 \\
    D-Wave(Ours) & \textcolor{black}{\textbf{0.84  $\pm$  0.104}} & 0.90  $\pm$  0.104 & 0.81  $\pm$  0.080 & \textcolor{gray}{\textbf{0.93  $\pm$  0.095}} & 0.87  $\pm$  0.096 \\
    \end{tabular}%
    }
  \label{tab:willow}\vspace{-3.5mm}
\end{table}%

\vspace{-2mm}\section{Experiments and Evaluations}
\label{sec:exp}\vspace{-2mm}

\parahead{Real dataset} 
To showcase that quantum computers offer a promising way to solve the challenging multi-view matching problems, we extract four categories (\emph{duck}, \emph{car}, \emph{winebottle}, \emph{motorbike}) of Willow Object  Classes~\cite{cho2013learning} composed of 40 RGB images each, acquired \emph{in the wild} (see Fig.~\ref{fig:willowcars}). The images suffer from significant pose, lighting and environment variation. Hence, ten keypoints are manually annotated on each image. We use the first four of these keypoints and create 35 small problems for each category by creating a fully connected graph composed of all four consecutive frames. We follow~\cite{wang2018multi} and extract local features from a set of $227\times 227$ patches centered around the annotated landmarks, using Alexnet~\cite{krizhevsky2012imagenet} pretrained on ImageNet~\cite{deng2009imagenet}. The feature map responses of \emph{Conv4} and \emph{Conv5} layers are then matched by the Hungarian algorithm~\cite{munkres1957algorithms} to initialize the synchronization. As the data is manually annotated, the ground-truth relative maps are known. 

\parahead{Synthetic dataset}
For a controlled evaluation of our method, similar to~\cite{birdal2019probabilistic}, we generate synthetic datasets composed of graphs with $m=|\Vertex|$ nodes and $|\Edge|=m(m-1)$ edges. At each node $i$ we generate $n$ points mappable by a permutation $\Pm_{ij}$ to $n$ other points at node $j$. Hence, all the permutation matrices $\Pm_i$ or $\Pm_{ij}$ are total \textit{i.e.,} of size $n\times n$. While the graph is by default fully connected, in certain experiments we randomly drop certain edges to have a completeness $C$ where $0.5<C<1$. For instance, $C=0.75$ means that only $75\%$ of the edges are actively present. We optionally perturb the relative permutations by swapping a percentage of the rows and columns. We call this the \emph{swap ratio} and denote it as $\sigma$, where $0\leq\sigma\leq 0.25$. 

\parahead{Evaluation methodology} 
We implement our algorithm on D-Wave Advantage $1.1$ and compare it against the state-of-the-art methods   MatchEIG~\cite{maset2017}, MatchALS~\cite{zhou2015multi}, MatchLift~\cite{huang2013consistent},  MatchBirkhoff~\cite{birdal2019probabilistic} as well as to the \emph{exhaustive} solution obtained by enumerating all possible permutations. 
In all of our evaluations we report the number of bits correctly detected and call this metric \emph{accuracy}. For synthetic evaluations involving noise, we generate seven random problems with the same $\sigma$ and average the results.

\insertimageC{1}{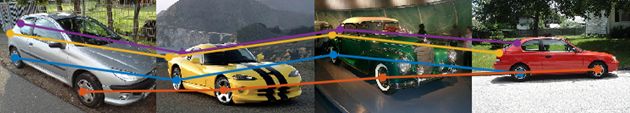}{A random \textit{car} example 
from Willow Object Classes~\cite{cho2013learning}.\vspace{-3mm}}{fig:willowcars}{t!}
\begin{figure}[b]
    \centering 
    \includegraphics[width=\linewidth, height= 85pt]{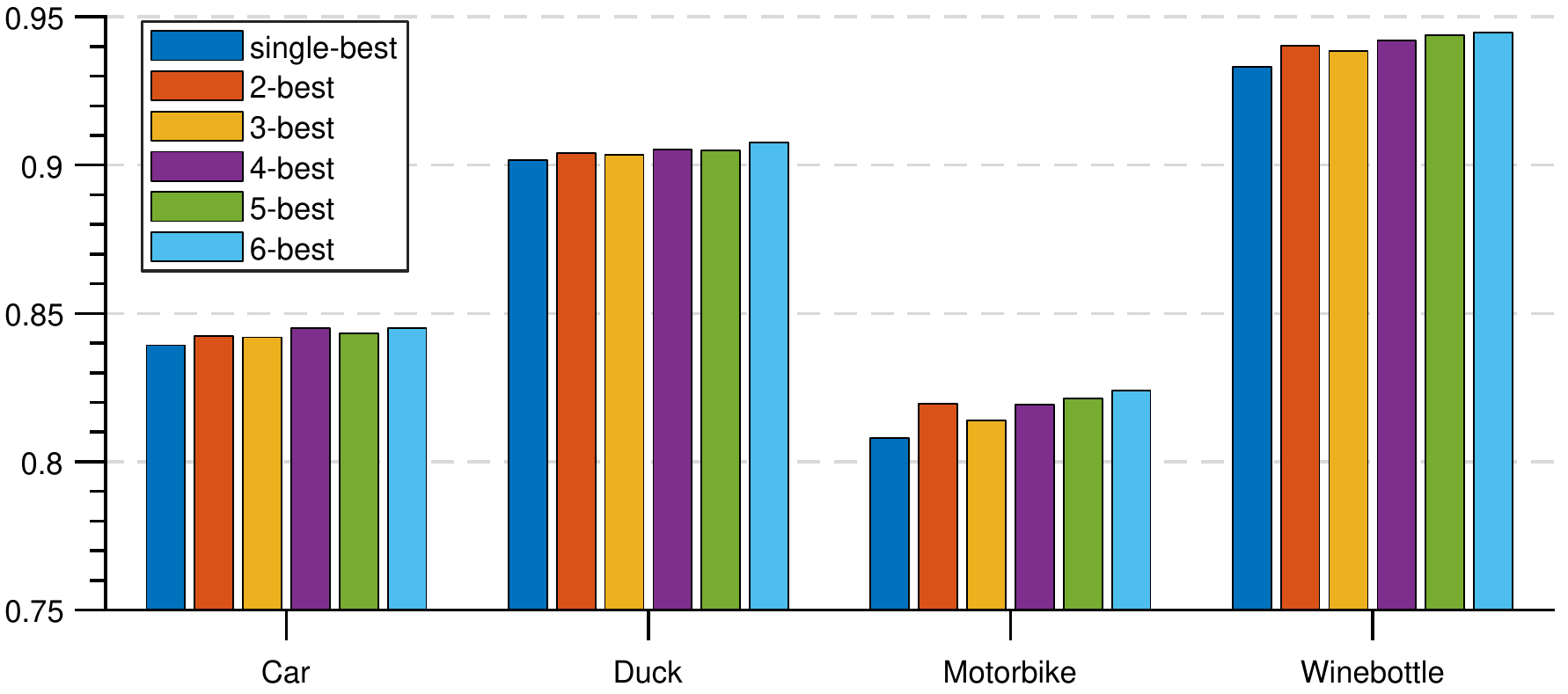} 
    \caption{Samples from the quantum annealer can be helpful in improving the solution quality or reporting  uncertainty.\vspace{-3mm}
    } 
    \label{fig:uncertainty} 
\end{figure} 
\subsection{Evaluations on D-Wave Advantage}\label{ssec:evaluations_Dwave} 
In this section, we describe our experiments on D-Wave Advantage system 1.1, which is an AQC with $5436$ qubits arranged on a graph of cells with eight qubits each. 
Every qubit operates under ${\approx}15.8 mK$ temperature and is connected to $15$ other qubits from the same or other cells which enables compact minor embeddings, \textit{i.e.,} mappings of target QUBOs to the processor topology with shorter chains of physical qubits compared to the previous $2000Q$ \cite{2020arXiv200300133B}. 
To map logical qubits which are connected to more than $15$ other logical qubits, chaining of physical qubits is necessary, \textit{i.e.,} entanglement between qubit states. 
These chains are maintained by auxiliary magnetic fields and can break during annealings. 
The mechanism of resolving broken chains is majority voting. 
Minor embeddings are performed automatically via~\cite{Cai2014}, and each annealing takes $20 {\mu}s$ -- apart from AQC problem transmission overheads which can sum up to $0.1$ sec, minor embedding time (which, however, can be pre-computed for multiple problem sizes) and AQC waiting time. 
We access D-Wave machines remotely through \textit{Leap2} \cite{DWave_Leap}. 
The total AQC runtime spent in the experiments amounts to ${\approx}15.5$ min. ($> 5 \cdot 10^5$ annealings in total). %
We show exemplary embeddings in our supplement. 

\parahead{Evaluations on the real dataset}
We begin by putting D-Wave to test on synchronizing real data.~\cref{tab:willow} shows the accuracy of the best (lowest energy) solution we attain for different categories, as well as averaged over all classes. It is seen that while quantum solution is not the top-performer, it is certainly on par with the well engineered approaches of the state of the art. Moreover, on a theoretical note, it enjoys polynomial speedup.
Note that the exhaustive solution performs the best on this dataset. This further encourages research on making better quantum computers reducing the gap between the true and the machine-computed optima. 

\parahead{Explaining the posterior} 
While not being a true Bayesian inference, quantum annealing can  provide samples from the energy landscape or the Hamiltonian. The samples could be used either, for instance, in associating a confidence to solutions or maybe in scenarios like active learning. To assess the usability of the samples, instead of looking into the lowest energy solution, we look at $k$-lowest energy solutions of the previous real data evaluation. We process them jointly and correct the erroneous bits by replacing  them with the most frequent ones over all the  samples.~\cref{fig:uncertainty} shows on the real  benchmark that for increasing $k$ values, the  accuracy also increases. This validates that  different samples from a quantum annealer can be  informative in {exploring} the energy landscape  \textit{i.e.,} posterior defined  in~\cref{sec:method}. 

\begin{figure}
    \centering
    \includegraphics[width=\linewidth]{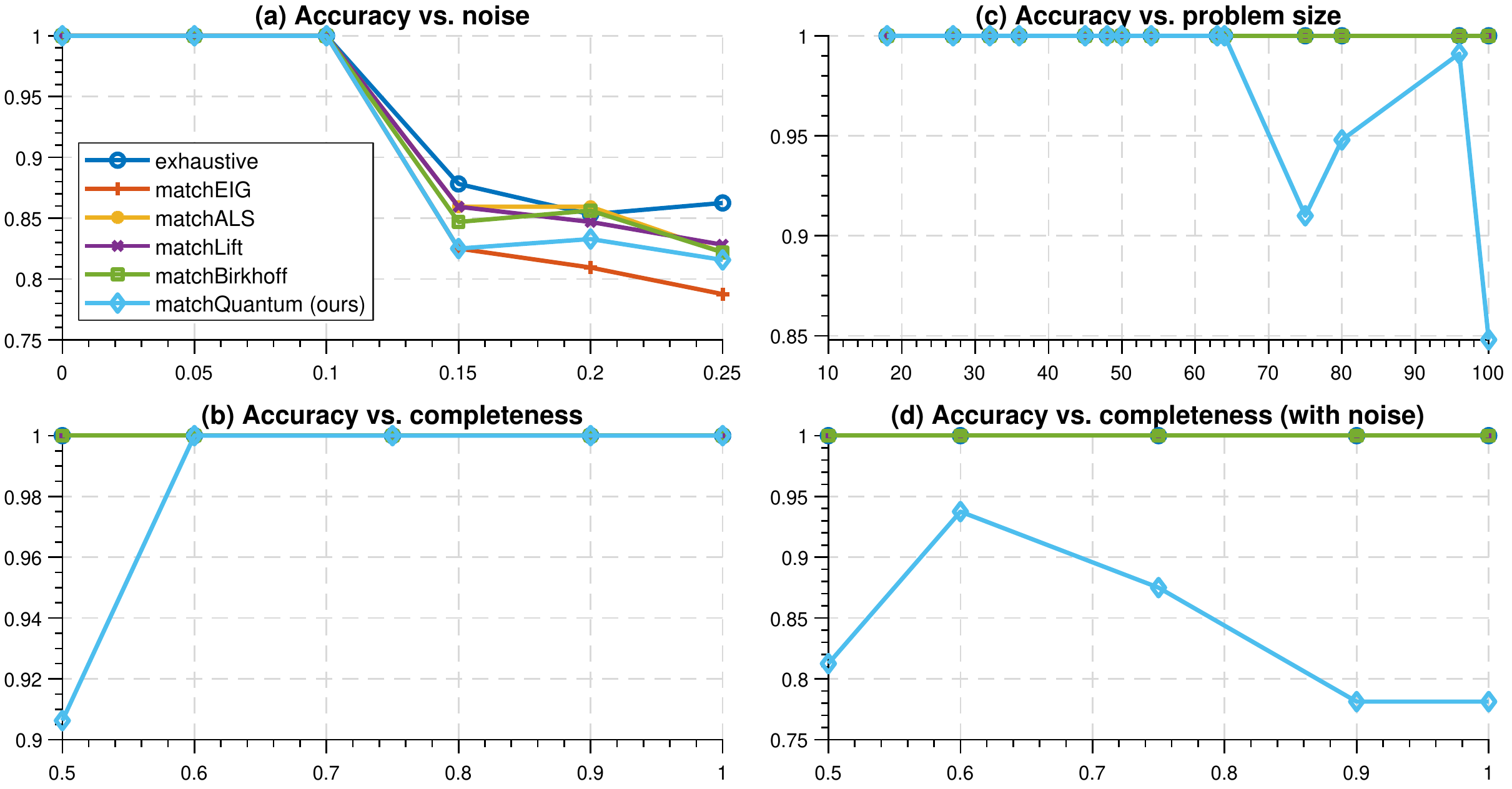} 
    \caption{Evaluations on the synthetic dataset (n = 4, m = 4).\vspace{-4mm}} 
    \label{fig:evaluations_synthetic} 
\end{figure} 
\parahead{Evaluations on the synthetic dataset} 
We now interrogate various characteristics of our quantum approach, \emph{QuantumSync}. Our results are plotted in~\cref{fig:evaluations_synthetic}. First, we inspect the behaviour under increasing noise.~\cref{fig:evaluations_synthetic}(a) shows that all the methods can handle low noise regimes. The increasing noise similarly impacts the methods we test. Our approach while not being superior to any, is on par.~\cref{fig:evaluations_synthetic}(b) shows that \textit{QuantumSync} is significantly impacted from the increased problem sizes ($mn^2$). This shows perhaps the most important limitation of current quantum computers, \textit{i.e.,} we can only reliably handle the problems of size $<64$. Note that, this is also the size of our real sub-Willow dataset. The rest of the plots (c,d) show the impact of connectivity (graph completeness) on the solution quality. With noise ($\sigma=0.1$ for this case) sparse graphs cause significant problems. 

\parahead{Parameter selection}
We now evaluate, with the help of our synthetic data, how the solution and its probability changes w.r.t.~the chain strength $\chi$ and $\lambda$. 
Both parameters strongly influence the probability to measure optimal solutions. 
$\chi$ being too high keeps the chains unbroken during annealings, but adversely affects the solutions. 
Too low $\chi$ lead to broken chains which, however, in many cases can be resolved with majority voting. 
Initial trial experiments help us to identify $\chi = 3.0$ and $\lambda = 2.5$ as optimal parameters which are kept fixed in our experiments. 
Note that the ablative study for $\lambda$ on a classical computer in~\cref{ssec:ablation_studies_classical_computer} also suggests $\lambda > 2.0$ as a suitable value for a range of $(n, m)$. 
~\cref{tab:experiment_DWave1} highlights the effect of varying $\chi$
in the tests with $n = 3, m = 3$ and $n=4, m = 4$ on Advantage $1.1$. 
We report the number of logical qubits of the target problem $n_l$, the number of physical qubits required to embed the problem on AQC $n_{ph}$, the average maximum chain length of the embedding $l$ and the average number of annealings leading to the optimal solution out of $200$ samples. 
Each test for each $\chi$ is repeated $50$ times. 

\begin{figure}[t]
    \centering 
    \includegraphics[width=\linewidth]{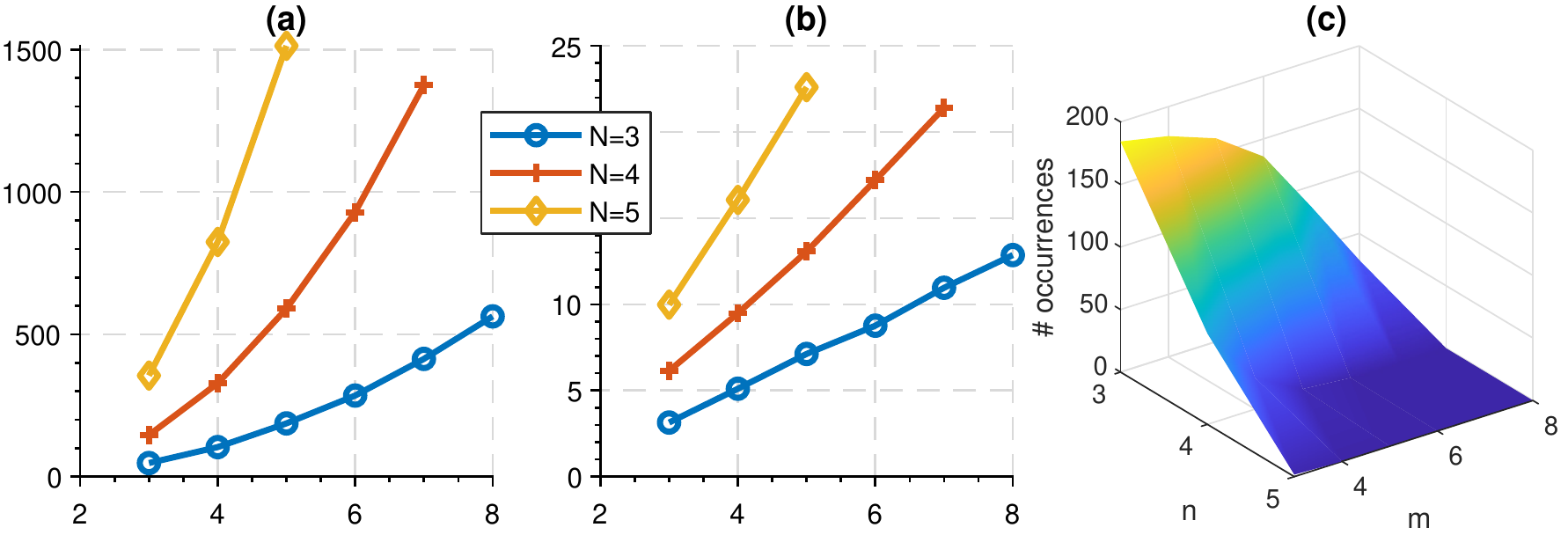} 
    \caption{For different $n$ and an increasing number of views $m$,  (\textbf{a}) plots the number of qubits required to map a problem; and (\textbf{b}) at $\chi = 3.0$, shows the required maximum chain length required to embed the problem on Advantage $1.1$. (\textbf{c}) plots the average number of measured optimal solutions in $200$ samples, for different pairs of $n$ and $m$ (averaged over $50$ repetitions).
    \vspace{-3mm}
    } 
    \label{fig:qubits} 
\end{figure}

\begin{table}[h] 
  \centering 
  \caption{
  The table summarises the effect of varying chain strength $\chi$ on the number of measured optimal solutions out of $200$, for $n=3, m=3$ (first row) and $n=4,m=4$ (second row). 
  } 
  \setlength{\tabcolsep}{3pt} 
  \resizebox{\columnwidth}{!} 
  { 
    \begin{tabular}{ccccccc} 
         $\boldsymbol{n_{l}}$ / ${\boldsymbol{n_{ph}}}$  / $\boldsymbol{l}$  & $\boldsymbol{\xi = 1}$ &  $\boldsymbol{\xi = 2}$ & $\boldsymbol{\xi = 3}$ &  $\boldsymbol{\xi = 4}$ & $\boldsymbol{\xi = 5}$ \\ 
    \cmidrule{1-6} 
    $18$ / $48$ / $3.2$ & $120.9$ $\pm$ $48.8$ & $187.8$ $\pm$ $11.5$ &  $180.9$ $\pm$ $11.2$ & $150.3$ $\pm$ $19.9$ & $85.3$ $\pm$ $23.6$ \\ 
    $48$ / $325$ / $9.5$ & $0.1$ $\pm$ $0.30$ & $9.9$ $\pm$ $10.08$ &  $17.3$ $\pm$ $14.0$ & $6.1$ $\pm$ $7.77$ & $1.42$ $\pm$ $2.72$ 
    \end{tabular} 
    }\vspace{-2mm} 
  \label{tab:experiment_DWave1}%
\end{table}%
\parahead{Different problem sizes and minor embedding} 
Next, we systematically analyze which problem sizes can be successfully embedded on Pegasus topology of Advantage system $1.1$ and solved globally optimally with high probability over $200$ samples. 
In each configuration of $n$ and $m$ which can be successfully solved on D-Wave, we report the average number of measurements corresponding to the global optimum and its standard deviation over $50$ repetitions with $200$ annealings each. 
We also analyze minor embedding in the same experiments and report the average number of physical qubits used in the embedding, the maximum chain lengths along with their standard deviations over $50$ runs. 
Fig.~\ref{fig:qubits} visualises the experimental outcomes. 
We see that with the increasing number of logical qubits, the number of physical qubits required for the embedding as well as the maximum chain length increase (Fig.~\ref{fig:qubits}-(a),(b)). 
For $n = 3$ and $m = 3$, the ratio is $c = \frac{n_{ph}}{n_l} \approx 2.65$. It increases to $c \approx 14.5$ for $n = 5$ and $m = 5$. 
This is not surprising, since longer qubit chains increase the probability of chain breaks and, hence, decrease the overall probability to measure optimal solution. 
Fig.~\ref{fig:qubits}-(c) shows the number of optimal measurements out of  $200$, for varying $n$ and $m$. 
We see that our \textit{QuantumSync} can solve the cases $n = 3, m = 8$, with probability to measure optimal solution in a single annealing $\rho = 4.39\%$; $n = 4, m = 4$, with $\rho = 12.5\%$, and $n = 5, m = 4$, with $\rho < 1\%$. 
At the same time, for problems with $n = 3$ and $m < 7$, $\rho > 62\%$. 
Note that certain problems are unmappable to D-Wave due to their size and thus we cannot report the statistics for those. 

\subsection{Ablation Studies on a Classical Computer}\label{ssec:ablation_studies_classical_computer} 
We now study the global minima of our problem on a classical computer and we design seven random problem instances that are globally solvable on standard hardware. Hence, we choose $n=3$, $m=|\Vertex|=3$ and $C=1$ (fully connected graph). For such a small size, we could exhaustively search for the global optimum both over binary variables and permutation matrices. Note that while the latter is the reasonable (actual) search space, an AQC can only optimize over the former. Hence, we are interested in quantifying the gap between the two and verify that our formulation indeed allows a QUBO-solver to achieve the global optimum for the problem at hand. 
\begin{figure}
    \centering
    \includegraphics[width=\linewidth, height=155pt]{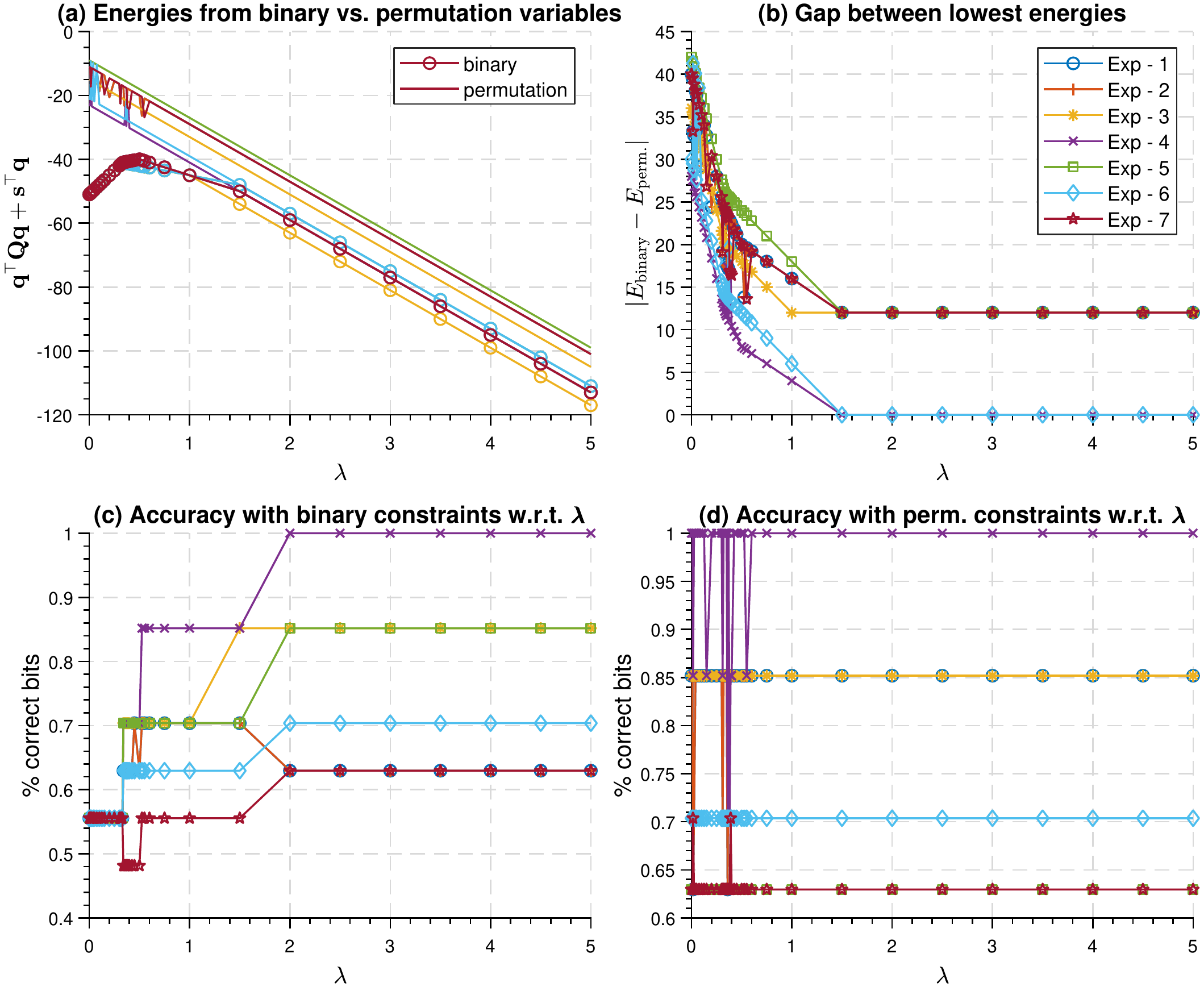}
    \caption{Studying the problem by exhaustive solutions on a classical computer. We solve seven random synthetic problems with $n=3$, $m=|\Vertex|=3$ and $\sigma=0.2$ by searching over \textbf{binary variables} (binary) or \textbf{permutation matrices} (perm.) exhaustively for the global optimum. We show: \textbf{(a)} the energies attained for both cases, \textbf{(b)} the gap (absolute difference) between the lowest energies in (a), \textbf{(c)} the ratio of correctly guessed bits as a function of the regularizer ($\lambda$) for the binary constraints, \textbf{(d)} same plot in (c) but for permutation matrices. Note, due to significant noise, the global optimum is not always the same as ground  truth.\vspace{-4mm}} 
    \label{fig:classical} 
\end{figure} 

\vspace{7pt}
\noindent\textbf{How can we choose the regularization coefficient {\large $\lambda$}?} 
The coefficient $\lambda$ is one of the most important hyper-parameters of our algorithm as it balances the data term \textit{vs} permutation penalty. Hence, we investigate its behavior. 
~\cref{fig:classical} shows that over those seven experiments where $\sigma=0.2$, while the lowest energies between the two solutions can differ (\textbf{a} and \textbf{b}), for a wide variety of $\lambda$-choices the accuracy attained by binary optimization and permutation optimization can be very comparable (\textbf{c} \textit{vs}  \textbf{d}). As long as $\lambda$ is not small (\textit{e.g.,} $<2$), we observe almost identical performance. This positive result has motivated us to settle for a single value $\lambda=2.5$ for all of our evaluations (including Sec.~\ref{ssec:evaluations_Dwave}). 

\vspace{7pt}
\noindent\textbf{Are permutation constraints effective?} 
As D-Wave cannot search over the permutations but only over binary variables, it is of interest to see whether our permutation-ness regularization really works. To investigate that, we design random experiments ($n=3, m=3$) with  increasing noise (swap ratio) where the GT is known. We then form the constrained $\Q$ matrix and solve it via exhaustive search on a classical computer. Averaged over seven experiments,~\cref{fig:energySR} shows that: (i) for low noise regime, optimizing over general binary variables $\bfB$ or over permutations $\Pset$ are indifferent and global optimum can always be found, (ii) for higher noise levels, while the energies attained seem to differ, the final accuracy is very similar. Hence, we conclude that injecting permutation constraints into $\Q$ as proposed is useful and makes it possible to use binary variables instead of permutations. This justifies why an adiabatic computer such as D-Wave could obtain global optima.
\begin{figure}
    \centering
    \includegraphics[width=\linewidth]{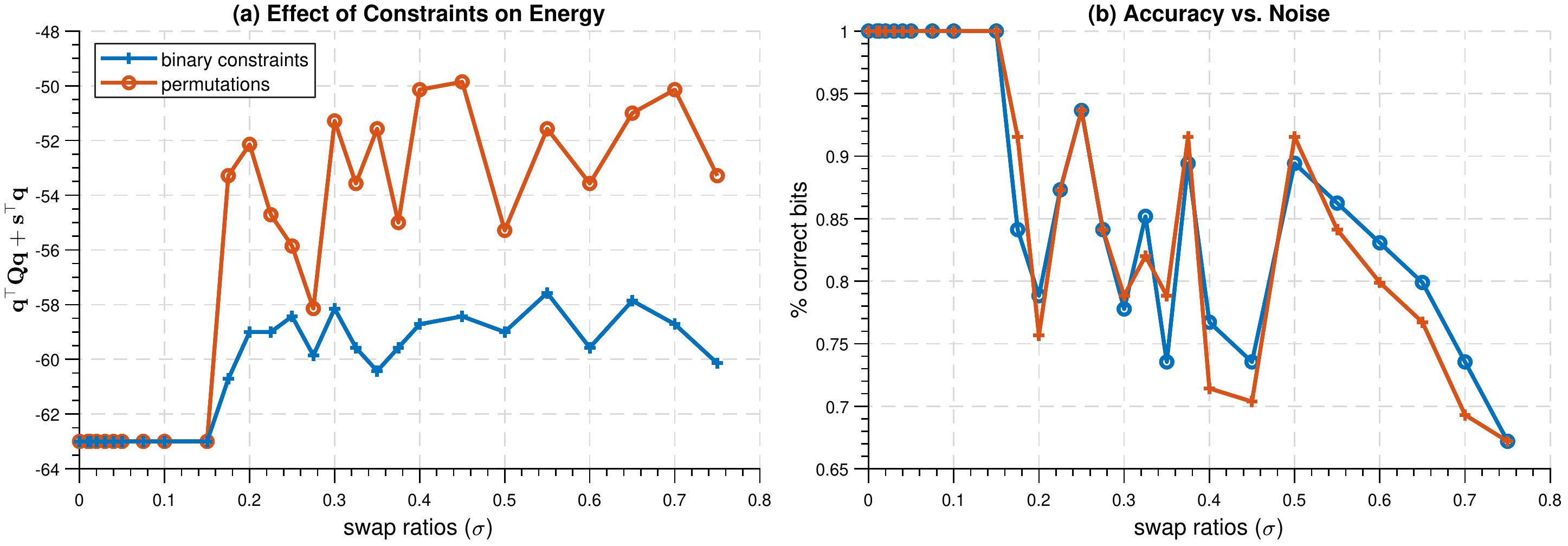}
    \caption{Effect of noise on the accuracy of the solution obtained by optimizing either over binary variables or permutation matrices. \textbf{(a)} Energy levels when the solution is restricted to binary variables or permutation matrices. \textbf{(b)} Accuracy attained by these two restricted solutions.\vspace{-3mm}}
    \label{fig:energySR} 
\end{figure}

\vspace{-3mm}\section{Conclusion} \vspace{-1mm}
We presented \emph{QuantumSync}, the first quantum approach to synchronization. We specifically focused on the group of permutations and showed how to formulate such problems for an adiabatic computer. We then used the cutting-edge quantum hardware to solve real-world problems with global guarantees. 
Our forward-looking experiments demonstrate that quantum computing hardware has reached the level that it can be applied to real-world problems with high potential to improve upon the known classical methods. 
We believe that our technique can inspire new generations of better algorithms for related and other computer vision problems, and we expect to see more work in the field in near future. 

{
\footnotesize
\noindent\textbf{Acknowledgements.} 
The work was supported by the ERC Grant 4DReply (770784), a Vannevar Bush Faculty fellowship, a grant from the Stanford-Ford Alliance and gifts from Amazon AWS and Snap, Inc.
}

{\small
\bibliographystyle{ieee_fullname}
\bibliography{biblio}

\begin{thebibliography}{100}\itemsep=-1pt

\bibitem{IBM}
Ibm q experience.
\newblock \url{https://quantum-computing.ibm.com}.
\newblock Accessed on the 17.01.2021.

\bibitem{DWAVE}
Technical description of the d-wave quantum processing unit.
\newblock \url{https://docs.dwavesys.com/docs/latest/doc_qpu.html}.
\newblock Accessed on the 17.01.2021.

\bibitem{Aimeur2013}
Esma A\"{\i}meur, Gilles Brassard, and S\'{e}bastien Gambs.
\newblock Quantum speed-up for unsupervised learning.
\newblock {\em Machine Learning}, 90(2):261–287, 2013.

\bibitem{Arrigoni2019}
Federica Arrigoni and Andrea Fusiello.
\newblock Synchronization problems in computer vision with closed-form
  solutions.
\newblock {\em International Journal of Computer Vision (IJCV)}, 2019.

\bibitem{arrigoni2017synchronization}
Federica Arrigoni, Eleonora Maset, and Andrea Fusiello.
\newblock Synchronization in the symmetric inverse semigroup.
\newblock In {\em International Conference on Image Analysis and Processing
  (ICIAP)}, pages 70--81, 2017.

\bibitem{arrigoni2016spectral}
Federica Arrigoni, Beatrice Rossi, and Andrea Fusiello.
\newblock Spectral synchronization of multiple views in se (3).
\newblock {\em SIAM Journal on Imaging Sciences}, 9(4):1963--1990, 2016.

\bibitem{arute2019quantum}
Frank Arute, Kunal Arya, Ryan Babbush, Dave Bacon, Joseph~C. Bardin, Rami
  Barends, Rupak Biswas, Sergio Boixo, Fernando G. S.~L. Brandao, David~A.
  Buell, et~al.
\newblock Quantum supremacy using a programmable superconducting processor.
\newblock {\em Nature}, 574(7779):505--510, 2019.

\bibitem{belov2019geometry}
Vadim Belov.
\newblock On geometry and symmetries in classical and quantum theories of gauge
  gravity.
\newblock {\em arXiv:1905.06931}, 2019.

\bibitem{bernard2015solution}
Florian Bernard, Johan Thunberg, Peter Gemmar, Frank Hertel, Andreas Husch, and
  Jorge Goncalves.
\newblock A solution for multi-alignment by transformation synchronisation.
\newblock In {\em Computer Vision and Pattern Recognition (CVPR)}, pages
  2161--2169, 2015.

\bibitem{bernard2018}
Florian Bernard, Johan Thunberg, Jorge Goncalves, and Christian Theobalt.
\newblock Synchronisation of partial multi-matchings via non-negative
  factorisations.
\newblock {\em Pattern Recognition}, 92:146 -- 155, 2019.

\bibitem{Berry2007}
{Dominic W.} Berry, Graeme Ahokas, Richard Cleve, and {Barry C.} Sanders.
\newblock Efficient quantum algorithms for simulating sparse hamiltonians.
\newblock {\em Communications in Mathematical Physics}, 270:359--371, 2007.

\bibitem{birdal2020synchronizing}
Tolga Birdal, Michael Arbel, Umut Simsekli, and Leonidas~J Guibas.
\newblock Synchronizing probability measures on rotations via optimal
  transport.
\newblock In {\em Computer Vision and Pattern Recognition (CVPR)}, pages
  1569--1579, 2020.

\bibitem{birdal2019generic}
Tolga Birdal, Benjamin Busam, Nassir Navab, Slobodan Ilic, and Peter Sturm.
\newblock Generic primitive detection in point clouds using novel minimal
  quadric fits.
\newblock {\em IEEE transactions on pattern analysis and machine intelligence},
  42(6):1333--1347, 2019.

\bibitem{birdalSimsekli2018}
Tolga {Birdal}, Umut {{\c S}im{\c s}ekli}, M.~Onur {Eken}, and Slobodan {Ilic}.
\newblock {Bayesian Pose Graph Optimization via Bingham Distributions and
  Tempered Geodesic MCMC}.
\newblock In {\em Advances in Neural Information Processing Systems (NeurIPS)},
  2018.

\bibitem{birdal2017cad}
Tolga Birdal and Slobodan Ilic.
\newblock Cad priors for accurate and flexible instance reconstruction.
\newblock In {\em International Conference on Computer Vision (ICCV)}, pages
  133--142, 2017.

\bibitem{birdal2019probabilistic}
Tolga Birdal and Umut Simsekli.
\newblock Probabilistic permutation synchronization using the riemannian
  structure of the birkhoff polytope.
\newblock In {\em Computer Vision and Pattern Recognition (CVPR)}, pages
  11105--11116, 2019.

\bibitem{BonehLipton1995}
Dan Boneh and Richard~J. Lipton.
\newblock Quantum cryptanalysis of hidden linear functions.
\newblock In {\em Advances in Cryptology (CRYPTO)}, 1995.

\bibitem{2020arXiv200300133B}
Kelly {Boothby}, Paul {Bunyk}, Jack {Raymond}, and Aidan {Roy}.
\newblock {Next-Generation Topology of D-Wave Quantum Processors}.
\newblock {\em arXiv e-prints}, 2020.

\bibitem{BornFock1928}
Max Born and Vladimir Fock.
\newblock Beweis des adiabatensatzes.
\newblock {\em Zeitschrift f{\"u}r Physik}, 51(3):165--180, 1928.

\bibitem{Boyda2017}
Edward Boyda, Saikat Basu, Sangram Ganguly, Andrew Michaelis, Supratik
  Mukhopadhyay, and Ramakrishna~R. Nemani.
\newblock Deploying a quantum annealing processor to detect tree cover in
  aerial imagery of california.
\newblock {\em PLoS ONE}, 12, 2017.

\bibitem{briales2017cartan}
Jesus Briales and Javier Gonzalez-Jimenez.
\newblock Cartan-sync: Fast and global se (d)-synchronization.
\newblock {\em IEEE Robotics and Automation Letters}, 2(4):2127--2134, 2017.

\bibitem{Brooke1999}
Justin~J. Brooke, David Bitko, Thomas~F. Rosenbaum, and Gabriel Aeppli.
\newblock Quantum annealing of a disordered magnet.
\newblock {\em Science}, 284(5415):779--781, 1999.

\bibitem{cadena2016past}
Cesar Cadena, Luca Carlone, Henry Carrillo, Yasir Latif, Davide Scaramuzza,
  Jos{\'e} Neira, Ian Reid, and John~J. Leonard.
\newblock Past, present, and future of simultaneous localization and mapping:
  Toward the robust-perception age.
\newblock {\em IEEE Transactions on robotics}, 32(6):1309--1332, 2016.

\bibitem{Cai2014}
Jun {Cai}, William~G. {Macready}, and Aidan {Roy}.
\newblock {A practical heuristic for finding graph minors}.
\newblock {\em arXiv e-prints}, 2014.

\bibitem{carlone2015initialization}
Luca Carlone, Roberto Tron, Kostas Daniilidis, and Frank Dellaert.
\newblock Initialization techniques for 3d slam: a survey on rotation
  estimation and its use in pose graph optimization.
\newblock In {\em International Conference on Robotics and Automation (ICRA)},
  pages 4597--4604, 2015.

\bibitem{chatterjee2013efficient}
Avishek Chatterjee and Venu~Madhav Govindu.
\newblock Efficient and robust large-scale rotation averaging.
\newblock In {\em International Conference on Computer Vision (ICCV)}, pages
  521--528, 2013.

\bibitem{chatterjee2017robust}
Avishek Chatterjee and Venu~Madhav Govindu.
\newblock Robust relative rotation averaging.
\newblock {\em IEEE transactions on pattern analysis and machine intelligence
  (TPAMI)}, 40(4):958--972, 2017.

\bibitem{chaudhury2015global}
Kunal~N. Chaudhury, Yuehaw Khoo, and Amit Singer.
\newblock Global registration of multiple point clouds using semidefinite
  programming.
\newblock {\em SIAM Journal on Optimization}, 25(1):468--501, 2015.

\bibitem{cho2013learning}
Minsu Cho, Karteek Alahari, and Jean Ponce.
\newblock Learning graphs to match.
\newblock In {\em International Conference on Computer Vision (ICCV)}, pages
  25--32, 2013.

\bibitem{Chuang1998}
Isaac~L. Chuang, Neil Gershenfeld, and Mark Kubinec.
\newblock Experimental implementation of fast quantum searching.
\newblock {\em Phys. Rev. Lett.}, 80:3408--3411, 1998.

\bibitem{Clader2013}
B.~David Clader, Bryan~C. Jacobs, and Chad~R. Sprouse.
\newblock Preconditioned quantum linear system algorithm.
\newblock {\em Phys. Rev. Lett.}, 110, 2013.

\bibitem{DWave_Leap}
{D-Wave Systems, Inc.}
\newblock Leap datasheet v10.
\newblock
  \url{https://www.dwavesys.com/sites/default/files/Leap_Datasheet_v10_0.pdf},
  2020.
\newblock online; latest access on the 8 November 2020.

\bibitem{Denchev2016}
Vasil~S. Denchev, Sergio Boixo, Sergei~V. Isakov, Nan Ding, Ryan Babbush, Vadim
  Smelyanskiy, John Martinis, and Hartmut Neven.
\newblock What is the computational value of finite-range tunneling?
\newblock {\em Phys. Rev. X}, 6:031015, 2016.

\bibitem{deng2020deep}
Haowen Deng, Mai Bui, Nassir Navab, Leonidas Guibas, Slobodan Ilic, and Tolga
  Birdal.
\newblock Deep bingham networks: Dealing with uncertainty and ambiguity in pose
  estimation.
\newblock {\em arXiv preprint arXiv:2012.11002}, 2020.

\bibitem{deng2009imagenet}
Jia Deng, Wei Dong, Richard Socher, Li-Jia Li, Kai Li, and Li Fei-Fei.
\newblock Imagenet: A large-scale hierarchical image database.
\newblock In {\em Computer Vision and Pattern Recognition (CVPR)}, 2009.

\bibitem{DiVincenzo2000}
David~P. DiVincenzo.
\newblock The physical implementation of quantum computation.
\newblock {\em Fortschritte der Physik}, 48(9‐11):771--783, 2000.

\bibitem{DonisVelaGarciaEscartin2018}
Alvaro Donis-Vela and Juan~Carlos Garcia-Escartin.
\newblock A quantum primality test with order finding.
\newblock {\em Quantum Info. Comput.}, 18(13–14):1143–1151, 2018.

\bibitem{Farhi2001}
Edward Farhi, Jeffrey Goldstone, Sam Gutmann, Joshua Lapan, Andrew Lundgren,
  and Daniel Preda.
\newblock A quantum adiabatic evolution algorithm applied to random instances
  of an np-complete problem.
\newblock {\em Science}, 292(5516):472--475, 2001.

\bibitem{Feynman1982}
Richard~P. Feynman.
\newblock Simulating physics with computers.
\newblock {\em International Journal of Theoretical Physics}, 21(6):467--488,
  1982.

\bibitem{fredriksson2012simultaneous}
Johan Fredriksson and Carl Olsson.
\newblock Simultaneous multiple rotation averaging using lagrangian duality.
\newblock In {\em Asian Conference on Computer Vision (ACCV)}, pages 245--258,
  2012.

\bibitem{gaitan2014graph}
Frank Gaitan and Lane Clark.
\newblock Graph isomorphism and adiabatic quantum computing.
\newblock {\em Physical Review A}, 89(2), 2014.

\bibitem{giridhar2006distributed}
Arvind Giridhar and Praveen~R. Kumar.
\newblock Distributed clock synchronization over wireless networks: Algorithms
  and analysis.
\newblock In {\em IEEE Conference on Decision and Control (CDC)}, pages
  4915--4920, 2006.

\bibitem{gojcic2020learning}
Zan Gojcic, Caifa Zhou, Jan~D Wegner, Leonidas~J Guibas, and Tolga Birdal.
\newblock Learning multiview 3d point cloud registration.
\newblock In {\em Computer Vision and Pattern Recognition (CVPR)}, pages
  1759--1769, 2020.

\bibitem{Golyanik_MBGA2020}
Vladislav {Golyanik}, Soshi {Shimada}, and Christian {Theobalt}.
\newblock Fast simultaneous gravitational alignment of multiple point sets.
\newblock In {\em International Conference on 3D Vision (3DV)}, 2020.

\bibitem{pointCorrespondence}
Vladislav Golyanik and Christian Theobalt.
\newblock A quantum computational approach to correspondence problems on point
  sets.
\newblock In {\em Computer Vision and Pattern Recognition (CVPR)}, 2020.

\bibitem{pointCorrespondenceSupp}
Vladislav Golyanik and Christian Theobalt.
\newblock A quantum computational approach to correspondence problems on point
  sets (supplementary material about the experiments on d-wave 2000q).
\newblock In {\em Computer Vision and Pattern Recognition (CVPR)}, 2020.

\bibitem{BHRGA2019}
Vladislav Golyanik, Christian Theobalt, and Didier Stricker.
\newblock Accelerated gravitational point set alignment with altered physical
  laws.
\newblock In {\em International Conference on Computer Vision (ICCV)}, 2019.

\bibitem{govindu2004lie}
Venu~Madhav Govindu.
\newblock Lie-algebraic averaging for globally consistent motion estimation.
\newblock In {\em Computer Vision and Pattern Recognition (CVPR)}, 2004.

\bibitem{govindu2014averaging}
Venu~Madhav Govindu and A Pooja.
\newblock On averaging multiview relations for 3d scan registration.
\newblock {\em IEEE Transactions on Image Processing}, 23(3):1289--1302, 2014.

\bibitem{Grover96afast}
Lov~K. Grover.
\newblock A fast quantum mechanical algorithm for database search.
\newblock In {\em Annual ACM Symposium on Theory of Computing}, 1996.

\bibitem{guibas2019condition}
Leonidas~J Guibas, Qixing Huang, and Zhenxiao Liang.
\newblock A condition number for joint optimization of cycle-consistent
  networks.
\newblock In {\em Advances in Neural Information Processing Systems (NeurIPS)},
  pages 1005--1015, 2019.

\bibitem{Hallgren2005}
Sean Hallgren.
\newblock Fast quantum algorithms for computing the unit group and class group
  of a number field.
\newblock In {\em Annual ACM Symposium on Theory of Computing}, page 468–474,
  2005.

\bibitem{Harrow2009}
Aram~W. Harrow, Avinatan Hassidim, and Seth Lloyd.
\newblock Quantum algorithm for linear systems of equations.
\newblock {\em Phys. Rev. Lett.}, 103, 2009.

\bibitem{hartley2011l1}
Richard Hartley, Khurrum Aftab, and Jochen Trumpf.
\newblock L1 rotation averaging using the weiszfeld algorithm.
\newblock In {\em Computer Vision and Pattern Recognition (CVPR)}, pages
  3041--3048, 2011.

\bibitem{hartley2013rotation}
Richard Hartley, Jochen Trumpf, Yuchao Dai, and Hongdong Li.
\newblock Rotation averaging.
\newblock {\em International Journal of Computer Vision (IJCV)}, 103(3), 2013.

\bibitem{hu2018distributable}
Nan Hu, Qixing Huang, Boris Thibert, and Leonidas Guibas.
\newblock Distributable consistent multi-object matching.
\newblock In {\em Computer Vision and Pattern Recognition (CVPR)}, 2018.

\bibitem{huang2019tensor}
Qixing Huang, Zhenxiao Liang, Haoyun Wang, Simiao Zuo, and Chandrajit Bajaj.
\newblock Tensor maps for synchronizing heterogeneous shape collections.
\newblock {\em ACM Transactions on Graphics (TOG)}, 38(4):106, 2019.

\bibitem{huang2013consistent}
Qi-Xing Huang and Leonidas Guibas.
\newblock Consistent shape maps via semidefinite programming.
\newblock In {\em Eurographics/ACM SIGGRAPH Symposium on Geometry Processing
  (SGP)}. Eurographics Association, 2013.

\bibitem{huang2019learning}
Xiangru Huang, Zhenxiao Liang, Xiaowei Zhou, Yao Xie, Leonidas~J Guibas, and
  Qixing Huang.
\newblock Learning transformation synchronization.
\newblock In {\em Computer Vision and Pattern Recognition (CVPR)}, pages
  8082--8091, 2019.

\bibitem{Jordan2012}
Stephen~P. Jordan, Keith S.~M. Lee, and John Preskill.
\newblock Quantum algorithms for quantum field theories.
\newblock {\em Science}, 336(6085):1130--1133, 2012.

\bibitem{KadowakiNishimori1998}
Tadashi Kadowaki and Hidetoshi Nishimori.
\newblock Quantum annealing in the transverse ising model.
\newblock {\em Phys. Rev. E}, 58:5355--5363, 1998.

\bibitem{Kane1998}
Bruce~E. Kane.
\newblock A silicon-based nuclear spin quantum computer.
\newblock {\em Nature}, 393:133–137, 1998.

\bibitem{Kassal2008}
Ivan Kassal, Stephen~P. Jordan, Peter~J. Love, Masoud Mohseni, and Al{\'a}n
  Aspuru-Guzik.
\newblock Polynomial-time quantum algorithm for the simulation of chemical
  dynamics.
\newblock {\em Proceedings of the National Academy of Sciences},
  105(48):18681--18686, 2008.

\bibitem{Kirkpatrick1983}
Scott Kirkpatrick, C.~Daniel Gelatt, and Mario~P. Vecchi.
\newblock Optimization by simulated annealing.
\newblock {\em Science}, 220(4598):671--680, 1983.

\bibitem{krizhevsky2012imagenet}
Alex Krizhevsky, Ilya Sutskever, and Geoffrey~E Hinton.
\newblock Imagenet classification with deep convolutional neural networks.
\newblock In {\em Advances in neural information processing systems (NeurIPS)},
  pages 1097--1105, 2012.

\bibitem{Lanting2014etal}
Trevor Lanting et~al.
\newblock Entanglement in a quantum annealing processor.
\newblock {\em Phys. Rev. X}, 4:021041, 2014.

\bibitem{Lanting2017}
Trevor Lanting, Andrew~D. King, Bram Evert, and Emile Hoskinson.
\newblock Experimental demonstration of perturbative anticrossing mitigation
  using nonuniform driver hamiltonians.
\newblock {\em Phys. Rev. A}, 96:042322, Oct 2017.

\bibitem{LeGall2012}
Francois Le~Gall.
\newblock Improved output-sensitive quantum algorithms for boolean matrix
  multiplication.
\newblock In {\em ACM-SIAM Symposium on Discrete Algorithms (SODA)}, pages
  1464--1476, 2012.

\bibitem{Lekitsch2017}
Bjoern Lekitsch, Sebastian Weidt, Austin~G. Fowler, Klaus M{\o}lmer, Simon~J.
  Devitt, Christof Wunderlich, and Winfried~K. Hensinger.
\newblock Blueprint for a microwave trapped ion quantum computer.
\newblock {\em Science Advances}, 3(2), 2017.

\bibitem{LiGhosh2020}
Junde Li and Swaroop Ghosh.
\newblock Quantum-soft qubo suppression for accurate object detection.
\newblock In {\em European Conference on Computer Vision (ECCV)}, 2020.

\bibitem{Lloyd2014}
Seth Lloyd, Masoud Mohseni, and Patrick Rebentrost.
\newblock Quantum principal component analysis.
\newblock {\em Nature Physics}, 10:631–633, 2014.

\bibitem{Manin1980}
Yuri Manin.
\newblock {\em Computable and Noncomputable}.
\newblock Sov. Radio., 1980.

\bibitem{maset2017}
Eleonora Maset, Federica Arrigoni, and Andrea Fusiello.
\newblock Practical and efficient multi-view matching.
\newblock In {\em International Conference on Computer Vision (ICCV)}, 2017.

\bibitem{Morello2010}
Andrea Morello, Jarryd~J. Pla, Floris~A. Zwanenburg, Kok~W. Chan, Kuan~Y. Tan,
  Hans Huebl, Mikko Möttönen, Christopher~D. Nugroho, Changyi Yang,
  Jessica~A. van Donkelaar, Andrew D.~C. Alves, David~N. Jamieson,
  Christopher~C. Escott, Lloyd C.~L. Hollenberg, Robert~G. Clark, and Andrew~S.
  Dzurak.
\newblock Single-shot readout of an electron spin in silicon.
\newblock {\em Nature}, 467:687–691, 2010.

\bibitem{munkres1957algorithms}
James Munkres.
\newblock Algorithms for the assignment and transportation problems.
\newblock {\em Journal of the society for industrial and applied mathematics},
  5(1):32--38, 1957.

\bibitem{Neukart2017}
Florian Neukart, Gabriele Compostella, Christian Seidel, David von Dollen,
  Sheir Yarkoni, and Bob Parney.
\newblock Traffic flow optimization using a quantum annealer.
\newblock {\em Frontiers in ICT}, 4:29, 2017.

\bibitem{Neven2012}
Hartmut Neven, Vasil~S. Denchev, Geordie Rose, and William~G. Macready.
\newblock Qboost: Large scale classifier training with adiabatic quantum
  optimization.
\newblock In {\em Asian Conference on Machine Learning (ACML)}, 2012.

\bibitem{nguyen2011optimization}
Andy Nguyen, Mirela Ben-Chen, Katarzyna Welnicka, Yinyu Ye, and Leonidas
  Guibas.
\newblock An optimization approach to improving collections of shape maps.
\newblock In {\em Computer Graphics Forum}, pages 1481--1491, 2011.

\bibitem{nguyen2020regression}
Nga T.~T. Nguyen, Garrett~T. Kenyon, and Boram Yoon.
\newblock A regression algorithm for accelerated lattice qcd that exploits
  sparse inference on the d-wave quantum annealer.
\newblock {\em Scientific Reports}, 10, 2020.

\bibitem{Ortiz2001}
Gerardo Ortiz, James~E. Gubernatis, Emanuel Knill, and Raymond Laflamme.
\newblock Quantum algorithms for fermionic simulations.
\newblock {\em Phys. Rev. A}, 2001.

\bibitem{OMalley2018}
Daniel O’Malley, Velimir~V. Vesselinov, Boian~S. Alexandrov, and Ludmil~B.
  Alexandrov.
\newblock Nonnegative/binary matrix factorization with a d-wave quantum
  annealer.
\newblock {\em PLOS ONE}, 13(12), 12 2018.

\bibitem{pachauri2013solving}
Deepti Pachauri, Risi Kondor, and Vikas Singh.
\newblock Solving the multi-way matching problem by permutation
  synchronization.
\newblock In {\em Advances in neural information processing systems}, pages
  1860--1868, 2013.

\bibitem{rempe2020caspr}
Davis Rempe, Tolga Birdal, Yongheng Zhao, Zan Gojcic, Srinath Sridhar, and
  Leonidas~J. Guibas.
\newblock Caspr: Learning canonical spatiotemporal point cloud representations.
\newblock In {\em Advances in Neural Information Processing Systems (NeurIPS)},
  2020.

\bibitem{rosen2019se}
David~M. Rosen, Luca Carlone, Afonso~S. Bandeira, and John~J. Leonard.
\newblock Se-sync: A certifiably correct algorithm for synchronization over the
  special euclidean group.
\newblock {\em The International Journal of Robotics Research},
  38(2-3):95--125, 2019.

\bibitem{RujikietgumjornCollins2013}
Sitapa {Rujikietgumjorn} and Robert~T. {Collins}.
\newblock Optimized pedestrian detection for multiple and occluded people.
\newblock In {\em Computer Vision and Pattern Recognition (CVPR)}, 2013.

\bibitem{salas2013slam++}
Renato~F. Salas-Moreno, Richard~A. Newcombe, Hauke Strasdat, Paul H.~J. Kelly,
  and Andrew~J. Davison.
\newblock Slam++: Simultaneous localisation and mapping at the level of
  objects.
\newblock In {\em Computer Vision and Pattern Recognition (CVPR)}, pages
  1352--1359, 2013.

\bibitem{schiavinato2017synchronization}
Michele Schiavinato and Andrea Torsello.
\newblock Synchronization over the birkhoff polytope for multi-graph matching.
\newblock In {\em International Workshop on Graph-Based Representations in
  Pattern Recognition}, pages 266--275, 2017.

\bibitem{Schuld2016}
Maria Schuld, Ilya Sinayskiy, and Francesco Petruccione.
\newblock Prediction by linear regression on a quantum computer.
\newblock {\em Phys. Rev. A}, 94, 2016.

\bibitem{SeelbachBenkner2020}
Marcel {Seelbach Benkner}, Vladislav {Golyanik}, Christian {Theobalt}, and
  Michael {Moeller}.
\newblock Adiabatic quantum graph matching with permutation matrix constraints.
\newblock In {\em International Conference on 3D Vision (3DV)}, 2020.

\bibitem{Shor1994}
Peter~W. {Shor}.
\newblock Algorithms for quantum computation: discrete logarithms and
  factoring.
\newblock In {\em Annual Symposium on Foundations of Computer Science}, 1994.

\bibitem{Simon1994}
Daniel~R. Simon.
\newblock On the power of quantum computation.
\newblock In {\em Annual Symposium on Foundations of Computer Science}, 1994.

\bibitem{simons1990overview}
Barbara Simons.
\newblock An overview of clock synchronization.
\newblock In {\em Fault-Tolerant Distributed Computing}. Springer, 1990.

\bibitem{singer2011angular}
Amit Singer.
\newblock Angular synchronization by eigenvectors and semidefinite programming.
\newblock {\em Applied and computational harmonic analysis}, 30(1):20--36,
  2011.

\bibitem{stollenwerk2019flight}
Tobias Stollenwerk, Elisabeth Lobe, and Martin Jung.
\newblock Flight gate assignment with a quantum annealer.
\newblock In {\em International Workshop on Quantum Technology and Optimization
  Problems}, 2019.

\bibitem{Suksmono2019}
Andriyan Suksmono and Yuichiro Minato.
\newblock Finding hadamard matrices by a quantum annealing machine.
\newblock {\em Scientific Reports}, 9, 12 2019.

\bibitem{TaShma2013}
Amnon Ta-Shma.
\newblock Inverting well conditioned matrices in quantum logspace.
\newblock In {\em ACM Symposium on Theory of Computing}, 2013.

\bibitem{thunberg2017distributed}
Johan Thunberg, Florian Bernard, and Jorge Goncalves.
\newblock Distributed methods for synchronization of orthogonal matrices over
  graphs.
\newblock {\em Automatica}, 80:243--252, 2017.

\bibitem{tron2014statistical}
Roberto Tron and Kostas Daniilidis.
\newblock Statistical pose averaging with non-isotropic and incomplete relative
  measurements.
\newblock In {\em European Conference on Computer Vision (ECCV)}, 2014.

\bibitem{tron2014distributed}
Roberto Tron and Rene Vidal.
\newblock Distributed 3-d localization of camera sensor networks from 2-d image
  measurements.
\newblock {\em IEEE Transactions on Automatic Control}, 59(12):3325--3340,
  2014.

\bibitem{vanDamShparlinski2008}
Wim {van Dam} and Igor~E. {Shparlinski}.
\newblock {Classical and Quantum Algorithms for Exponential Congruences}.
\newblock {\em arXiv e-prints}, 2008.

\bibitem{wang2013exact}
Lanhui Wang and Amit Singer.
\newblock Exact and stable recovery of rotations for robust synchronization.
\newblock {\em Information and Inference: A Journal of the IMA}, 2(2):145--193,
  2013.

\bibitem{wang2018multi}
Qianqian Wang, Xiaowei Zhou, and Kostas Daniilidis.
\newblock Multi-image semantic matching by mining consistent features.
\newblock In {\em Computer Vision and Pattern Recognition (CVPR)}, 2018.

\bibitem{Wesenberg2009}
Janus~H. Wesenberg, Arzhang Ardavan, G.~Andrew~D. Briggs, John J.~L. Morton,
  Robert~J. Schoelkopf, David~I. Schuster, and Klaus M\o{}lmer.
\newblock Quantum computing with an electron spin ensemble.
\newblock {\em Phys. Rev. Lett.}, 103, 2009.

\bibitem{WoottersZurek1982}
William~K. {Wootters} and Wojciech~H. {Zurek}.
\newblock {A single quantum cannot be cloned}.
\newblock {\em Nature}, 299(5886):802--803, 1982.

\bibitem{yu2016globally}
Jin-Gang Yu, Gui-Song Xia, Ashok Samal, and Jinwen Tian.
\newblock Globally consistent correspondence of multiple feature sets using
  proximal gauss--seidel relaxation.
\newblock {\em Pattern Recognition}, 51:255--267, 2016.

\bibitem{zach2010disambiguating}
Christopher Zach, Manfred Klopschitz, and Marc Pollefeys.
\newblock Disambiguating visual relations using loop constraints.
\newblock In {\em Computer Vision and Pattern Recognition (CVPR)}, pages
  1426--1433, 2010.

\bibitem{Zalka1998}
Christof Zalka.
\newblock Efficient simulation of quantum systems by quantum computers.
\newblock {\em Fortschritte der Physik}, 46(6‐8):877--879, 1998.

\bibitem{zhang2019path}
Zaiwei Zhang, Zhenxiao Liang, Lemeng Wu, Xiaowei Zhou, and Qixing Huang.
\newblock Path-invariant map networks.
\newblock In {\em Computer Vision and Pattern Recognition (CVPR)}, pages
  11084--11094, 2019.

\bibitem{Zhao2019}
Zhikuan Zhao, Alejandro Pozas-Kerstjens, Patrick Rebentrost, and Peter Wittek.
\newblock Bayesian deep learning on a quantum computer.
\newblock {\em Quantum Machine Intelligence}, pages 41--51, 2019.

\bibitem{zhou2015multi}
Xiaowei Zhou, Menglong Zhu, and Kostas Daniilidis.
\newblock Multi-image matching via fast alternating minimization.
\newblock In {\em International Conference on Computer Vision (ICCV)}, pages
  4032--4040, 2015.

\bibitem{zick2015experimental}
Kenneth~M. Zick, Omar Shehab, and Matthew French.
\newblock Experimental quantum annealing: case study involving the graph
  isomorphism problem.
\newblock {\em Scientific reports}, 5(1), 2015.

\bibitem{Zwanenburg2013}
Floris~A. Zwanenburg, Andrew~S. Dzurak, Andrea Morello, Michelle~Y. Simmons,
  Lloyd C.~L. Hollenberg, Gerhard Klimeck, Sven Rogge, Susan~N. Coppersmith,
  and Mark~A. Eriksson.
\newblock Silicon quantum electronics.
\newblock {\em Rev. Mod. Phys.}, 85:961--1019, 2013.

\end{thebibliography}
}

\onecolumn
\setcounter{section}{0}
\renewcommand\thesection{\Alph{section}}
\newcommand{\suppsection}{\subsection}
\clearpage
\begin{center}
\textbf{\Large Quantum Permutation Synchronization\\ ---Supplementary Material---}
\end{center}
\makeatletter

This part supplements our main paper by providing (i) proofs of the theorems and propositions contained in the main paper; (ii) further insights into the constraints that we propose; (iii) visualizations of minor embeddings on D-Wave; (iv) details and illustrations of the synthetic data; and (v) further descriptions of our real dataset.
The figures and tables introduced in this document are referenced using Roman numerals, 
whereas the original references to the main matter are preserved.

\section{Proof of Theorems}
\begin{proof}[Proof of Theorem \ref{theorem:cycle_consistency}.]
The proof follows from plugging Eq.~\eqref{eq:cc} into the definition of a null-cycle in Eq.~\eqref{eq:nullcycle} of the main paper, obtaining:
\begin{align}
f_c &= f_{1,2}\circ f_{2,3}\circ\cdots\circ f_{(n-1),n} \circ f_{n,1} \\
&= \big(f_1\circ f_2^{-1}\big)\circ\big(f_2\circ f_3^{-1}\big)\circ \cdots \circ \big(f_n\circ f_1^{-1}\big)\\
&= f_1 \circ \big(f_2^{-1}\circ f_2\big)\circ \big(f_3^{-1}\circ f_3\big) \circ \cdots \circ \big(f_n^{-1}\circ f_n\big)\circ f_1^{-1} \\ 
&= f_1 \circ f_1^{-1} = f_{null} \quad \forall c\in\Cycles. 
\end{align}
\end{proof}
\begin{proof}[Proof of Theorem \ref{theorem:gauge_freedom}.]
Without loss of generality, for an edge $(i,j)$, we consider the transformed cycle consistency constraint:
\begin{equation}
(f_i\circ f_g) \circ (f_j\circ f_g)^{-1} = (f_i\circ f_g) \circ (f_g^{-1} \circ f_j) = f_i\circ f_j^{-1}.
\end{equation}
This shows that consistency relation is unchanged under the action of an arbitrary element in the group $f_g$. 
\end{proof}

\section{Extended Proof of Proposition \ref{prop1}}
Here we provide more steps to the derivation in Prop.~\ref{prop1} of the main paper:
\begin{align}
\X^\star &= \argmin\limits_{\{\Pm_i \in \Pset_n^N\}} \sum_{(i,j)\in \Edge} \|\Pm_{ij} - \Pm_i \Pm_j^\top \|^2_\mathrm{F} \\
&= \argmin\limits_{\{\Pm_i \in \Pset_n^N\}} \sum_{(i,j)\in \Edge} \|\Pm_{ij}\|^2_\mathrm{F} - \|\Pm_i \Pm_j^\top \|^2_\mathrm{F} - 2\mathrm{tr}(\Pm_j\Pm_i^\top\Pm_{ij})\nonumber\\
&= \argmin\limits_{\{\Pm_i \in \Pset_n^N\}} \,2N^2 n-2\sum_{(i,j)\in \Edge} \mathrm{tr}(\Pm_j\Pm_i^\top\Pm_{ij})\\
&=\argmin\limits_{\{\Pm_i \in \Pset_n^N\}} \,-\sum_{(i,j)\in \Edge} \mathrm{tr}(\Pm_j\Pm_i^\top\Pm_{ij})\\
&=\argmin\limits_{\{\Pm_i \in \Pset_n^N\}} \,-\sum_{(i,j)\in \Edge} \vecm(\Pm_i)^\top (\Id \otimes \Pm_{ij} )\vecm(\Pm_j) \nonumber \\
&=\argmin\limits_{\{\Pm_i \in \Pset_n^N\}} \,-\sum_{(i,j)\in \Edge} \q_i^\top (\Id \otimes \Pm_{ij}) \q_j
=\argmin\limits_{\{\Pm_i \in \Pset_n^N\}} \, \q^\top \Q^\prime \q.
\end{align}

\section{Proof of Proposition \ref{prop2}}
Linear constraints can be injected into a QUBO problem in the following manner:
\begin{align}
\x^\star &= \argmin\limits_{\x \in \bfB} \, \x^\top \Q^\prime \x + \lambda \|\Ag\x-\bg\|_2^2\\
&= \argmin\limits_{\x \in \bfB} \, \x^\top \Q^\prime \x + \lambda \big(\Ag\x-\bg\big)^\top \big(\Ag\x-\bg\big)\\ 
&= \argmin\limits_{\x \in \bfB} \, \x^\top \big(\Q^\prime +\lambda\Ag^\top\Ag \big) \x - 2\lambda\bg^\top \Ag\x \\
&= \argmin\limits_{\x \in \bfB} \, \x^\top \Q \x + \s^\top\x.
\end{align}

\section{Example Constraint Matrices and Sparsity Patterns}
We have shown in the main paper that the permutation constraints can be formulated as a set of linear systems $\{\Ag_{i}\q_i=\bg\triangleq \one\}$. We now show concrete examples in two and three dimensions of these constraint matrices:
\begin{enumerate}[parsep=0pt,partopsep=0pt]
    \item $n=2$:
    \begin{equation}
    \Ag_i= \begin{bmatrix}
1 &1&0&0 \\
0&0&1&1 \\
1&0&1&0 \\
0&1&0&1
    \end{bmatrix}, \quad \bg_i = \one =  \begin{bmatrix}
    1\\
    1\\
    1\\
    1
    \end{bmatrix}.
\end{equation}
\item $n=3$:
\begin{equation}
    \Ag_i= \begin{bmatrix}
1 &1&1&0&0&0&0&0&0 \\
0&0&0&1&1&1&0&0&0 \\
0&0&0&0&0&0 &1&1&1 \\
1&0&0&1&0&0&1&0&0 \\
0&1&0&0&1&0&0&1&0 \\
0&0&1&0&0&1&0&0&1 \\
    \end{bmatrix}, \bg_i=  \begin{bmatrix}
    1\\
    1\\
    1\\
    1\\
    1\\
    1
    \end{bmatrix}.
\end{equation}
\end{enumerate}

\insertimageStar{1}{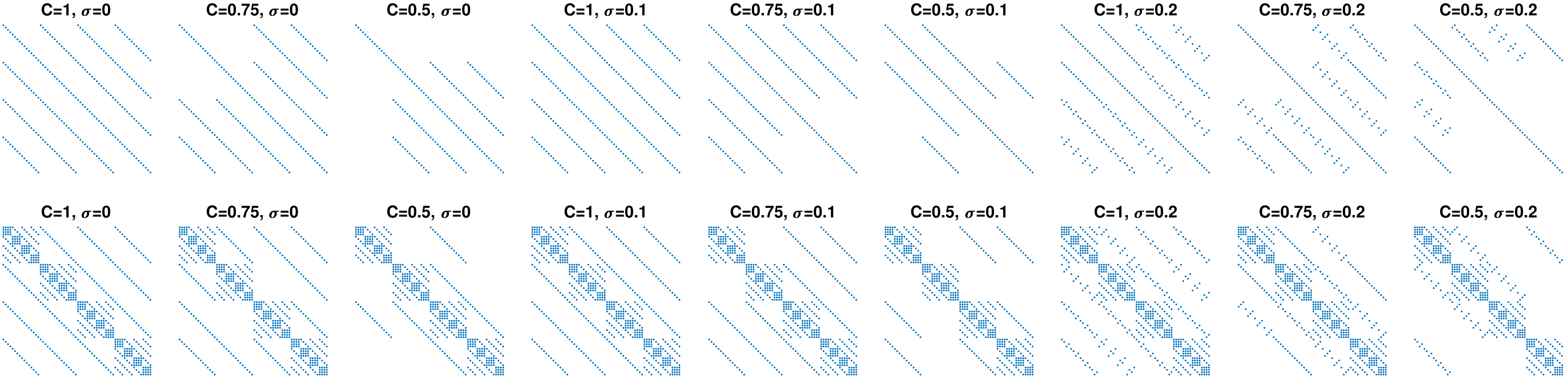}{Sparsity patterns of the unconstrained (\textbf{top row}) and constrained (\textbf{bottom row}) $\Q$-matrices. For the former we use $\lambda=0$ while the latter uses $\lambda=2$. The matrix is constructed for $n=m=4$ and thus has $mn^2=64$ elements per dimension. We visualize the matrices for a range of completeness $C$ and swap-ratio $\sigma$ values.\vspace{-3mm}}{fig:Q}{t!}

We further analyze how our quadratic constraint matrix looks like when these linear constraints are added. For different random experiments, we plot  in~\cref{fig:Q} these $\Q$-matrices for different values of completeness $C$ and swap ratio $\sigma$ when $n=m=4$. Note that this is a typical setting for our synthetic evaluations.

\section{Minor Embeddings and  Sampling}\label{app:minor_embeddings} 
In the default mode, \textit{Leap2} allocates couplings  between qubits in the minor embeddings even if they are equal to zero. 
If zero couplings are explicitly avoided, the total number of physical qubits required to minor-embed a logical problem, along with the maximum required chain length, decreases, and the probability of measuring the optimal solution in single annealing is increasing. 
See~\cref{fig:qubits_new} for the variant of~\cref{fig:qubits} of  the main paper, without allocated zero qubit couplings. 
We observe empirically that more compact minor embeddings ---in contrast to what one could presume--- do not allow to solve larger problem instances.
In~\cref{fig:minor_embeddings}, we show three exemplary minor embeddings on the Pegasus architecture of Advantage system $1.1$ for problems of different sizes. 
While the number of logical qubits for the shown embeddings ranges from $18$ (A) to $63$ (C), the number of the corresponding physical qubits ranges from $49$ (A) to $550$ (C). 
Note that the shown visualizations include zero couplings,  and more compact minor embeddings are possible  (\textit{cf.} Fig.~\ref{fig:qubits_new}). 
With the increasing problem size, we also see that maximum lengths of qubit chains $l$ in the embeddings, which encode the same logical qubit, increase as well. 
While $l = 4$ for $n = 3, m = 4$, it increases to $10$ and $13$ for the combinations  $n = 4, m = 4$ and $n = 3, m = 8$, respectively. 
For the same problems, AQC of the previous generation 2000Q requires $108$, $782$ or $1378$ physical qubits on average over $50$ runs, respectively.

\begin{figure*}[t]
    \centering 
    \includegraphics[width=\linewidth]{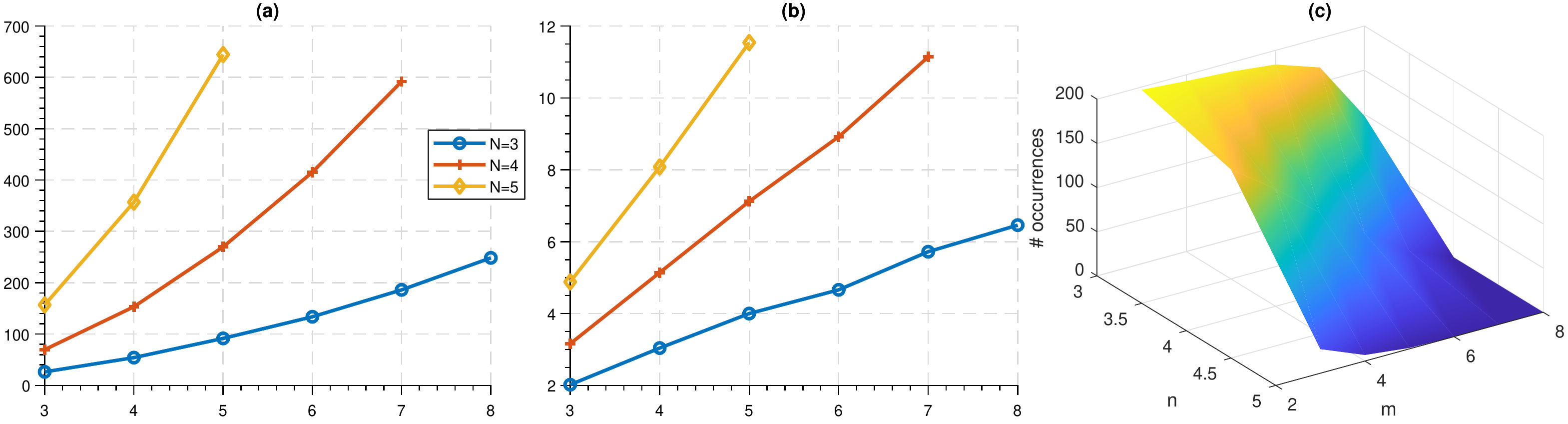} 
    \caption{This figure supplements Fig.~\ref{fig:qubits} of the main paper with the difference that the zero qubit couplings are excluded. For different $n$ and an increasing number of views $m$,  (\textbf{a}) plots the number of qubits required to map a problem; and (\textbf{b}) at $\chi = 3.0$, shows the required maximum chain length required to embed the problem on Advantage $1.1$. (\textbf{c}) plots the average number of measured optimal solutions in $200$ samples, for different pairs of $n$ and $m$ (averaged over $50$ repetitions).
    } 
    \label{fig:qubits_new} 
\end{figure*}

\parahead{Selecting qubit biases} 
We empirically find that the qubit biases have to be set differently compared to as derived in \eqref{eq:unconstrained_QUBO}, see~\cref{tab:experiment_App_E} for the summary of our experiment. 
Thus, $-\sqrt{|\s_k|}$ worked well (instead of the derived $-\operatorname{diag}({\bfQ})_k -\s_k$)\footnote{Our experiments suggest that there is a broader range of smaller biases which work as well as $-\sqrt{|\s_k|}$. 
In several our test cases, $-0.5 \s_k$ were leading to similar solution distributions. 
Further study is required on the differences between the  derived weights and the weights which should be set in practice.} in combination with the selected chain strengths and $\lambda$, see~\cref{ssec:evaluations_Dwave} for details on parameter selection. 
The possible reasons for that lie on the hardware side. 
First, the qubit biases and couplings are converted to magnetic fields acting on qubits, \textit{i.e.,} the weight encoding and annealing are analogue physical processes. 
The embedded problem is predominantly defined in terms of qubit couplings ($\frac{k^2 - k}{2}$ logical couplings at most in contrast to $k$ logical biases at most for a problem with $k$ logical qubits). 
Moreover, magnetic fields are imposed to keep chains of physical qubits intact. 
Last but not least, the range of real (floating-point) values which can be mapped to the native Ising format of D-Wave is limited (currently, it is $[-2; 2]$ for biases and $[-1; 1]$ for couplings on 2000Q and system $1.1$), and analytically derived biases and couplings have to be scaled down to the supported ranges. 
These factors can lead to offsets in qubit biases which are difficult to predict theoretically, as QUBO formulations are often derived without consideration of the minor embedding to a real AQC. 

\begin{figure}[!t]
    \centering 
    \includegraphics[width=\linewidth]{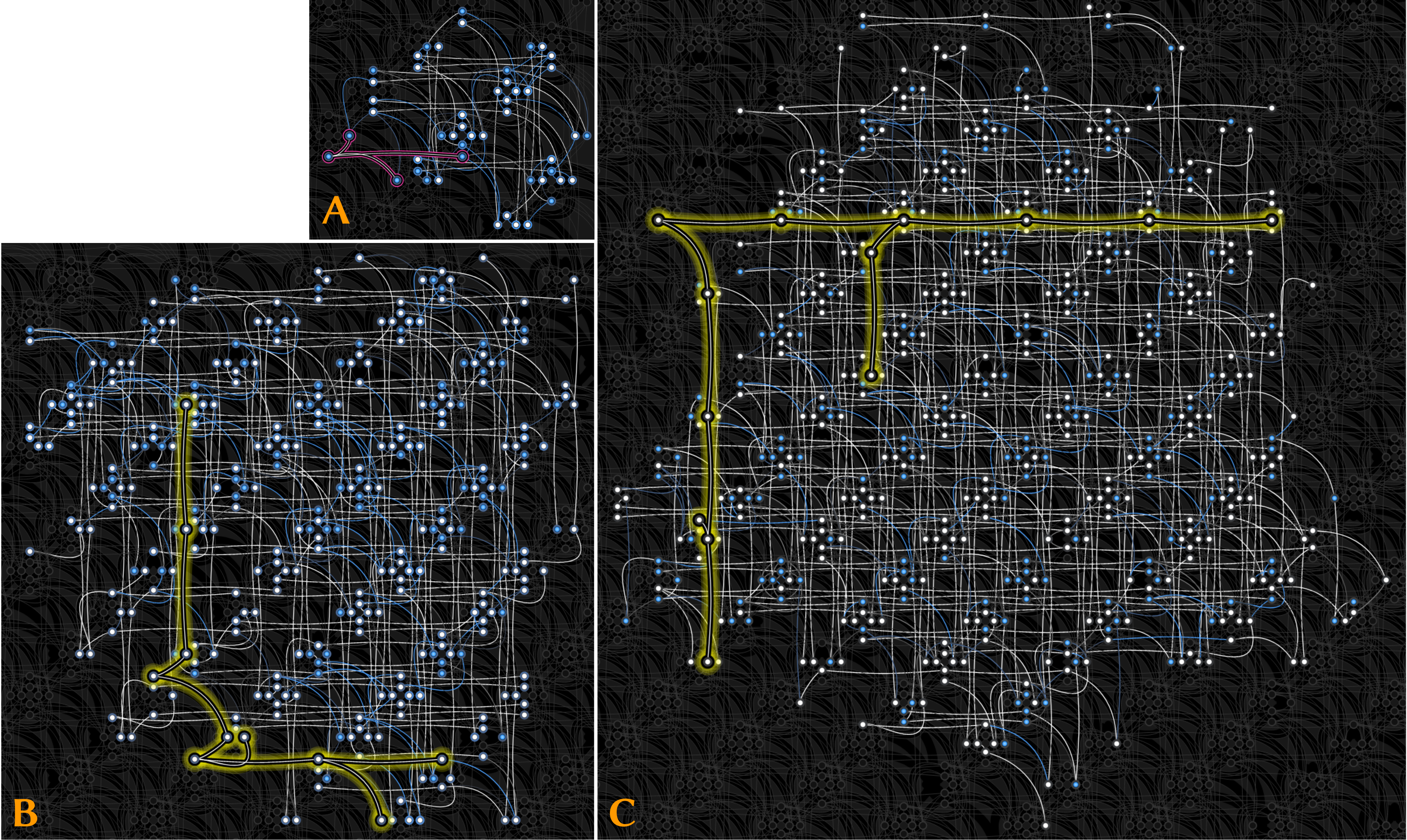} 
    \caption{ 
    Exemplary minor embeddings in the experiments with $n = 3, m = 3$ (A,  $n_{ph} = 49$), $n = 4, m = 4$ (B, $n_{ph} = 341$) and $n = 3, m = 8$ (C, $n_{ph} = 550$). 
    In each case, we highlight qubit chains of the maximum chain length in the embedding, either in magenta (A, no warnings) or yellow (B and C, chain length warnings). 
    Note that the shown minor embeddings  include zero qubit couplings. 
    } 
    \label{fig:minor_embeddings} 
\end{figure}

\begin{table}
  \centering 
  \caption{ 
  The table summarises the average number (and the corresponding standard deviation) of measured optimal solutions out of $200$, for $n=4, m=4$ and varying chain strength $\chi$, strictly as derived in~\cref{prop2} (first row) and with adjusted weights as described in~\cref{app:minor_embeddings}, \textit{i.e.,} $-\sqrt{|\s_k|}$ (second row). 
  Each number is reported for $50$ repetitions of the experiment. 
  } 
  \setlength{\tabcolsep}{3pt} 
  \resizebox{0.59\columnwidth}{!} 
  { 
    \begin{tabular}{cccccc} 
          & $\boldsymbol{\xi = 1}$ &  $\boldsymbol{\xi = 2}$ & $\boldsymbol{\xi = 3}$ &  $\boldsymbol{\xi = 4}$ & $\boldsymbol{\xi = 5}$ \\ 
    \cmidrule{1-6} 
    derived    & $0.02$ $\pm$ $0.14$ & $1.26$ $\pm$ $2.1$ & $1.38$ $\pm$ $1.75$ & $0.08$ $\pm$ $0.27$ & $0.0$ $\pm$ $0.0$\\ 
    adjusted   & $0.1$ $\pm$ $0.30$ & $9.9$ $\pm$ $10.08$ & $17.3$ $\pm$ $14.0$ & $6.1$ $\pm$ $7.77$ & $1.42$ $\pm$ $2.72$ 
    \end{tabular} 
    }\vspace{-2mm} 
  \label{tab:experiment_App_E}%
\end{table}%

\parahead{On sampling}
Note that, theoretically the samples coming from the quantum annealer can not be directly interpreted as the samples from the induced Bolzmann distribution characterized by $\beta$ (\textit{e.g.,} lower values of $\beta$ result in samples less constrained to the lowest energy states), as the annealer samples a modified posterior. Optionally, one could steer the samples towards the local minima using a CPU-based descent algorithm. However, we found that in practice this hack, also suggested by D-Wave, does not work well for capturing a diverse set of modes. Nevertheless, as we show in~\cref{fig:uncertainty} (in the main paper) using the original samples as alternative plausible solutions could still boost the accuracy albeit incrementally.
While a true posterior adjustment is not yet implemented in D-Wave, we foresee that upcoming years would witness a leap on these fronts. 

\section{Further Details on Synthetic Experiments}
We visualize in~\cref{fig:synth_data} some examples from our random synthetic dataset. While the common case of $n=4$ and $m=4$ can be visualized this way, the same does not hold true for other $n,m$ combinations. That is why the points on the grid only assist the visualization, and are not used in practice. The important cues are the correspondences denoting permutations, whose rows might be randomly swapped to inject noise.

\insertimageStar{1}{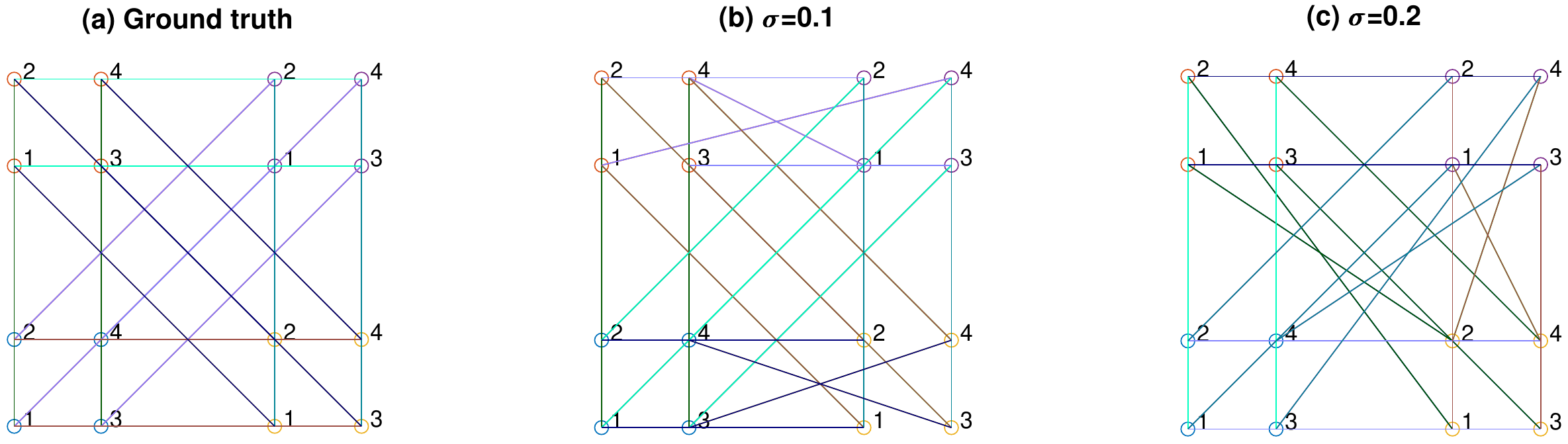}{Samples from our fully connected synthetic dataset for different values of swap ratio $\sigma$. For sparser graphs, where completeness is less than one, it is possible to imagine edges being dropped from these graphs. In the figure, each group (indicated using differently colored points) corresponds to a \emph{view} and each inter-group correspondence corresponds to a permutation that we optimize for. Note that the points are drawn as a grid to ease visual perception. Neither our algorithm nor the state-of-the-art methods we compare would use this information.\vspace{-3mm}}{fig:synth_data}{h!}

\section{Evaluations on the Real Dataset Based on the Willow Object Classes~\cite{cho2013learning}}

\parahead{Dataset description} 
Willow dataset is composed of 40 RGB images of four object classes: \emph{duck}, \emph{car}, \emph{winebottle}, \emph{motorbike} and \emph{face}. These images are extracted either from \emph{Caltech-256} or \emph{PASCAL VOC} datasets. 
The \emph{in-the-wild} nature of these images as shown in Figs.~\ref{fig:willow} and~\ref{fig:willow_matches} makes simple template matching ill-suited for correspondence estimation. This is why, for our evaluations---as explained in the main paper---we followed the same design as Wang~\textit{et al.}~\cite{wang2018multi} and benefited from the state-of-the-art-deep learning approaches. 
The authors of~\cite{cho2013learning} have manually selected ten distinctive points on the target object (category). These keypoints are annotated consistently across all the instances in each of the categories as we show in~\cref{fig:willow}. To this end, a semi-automatic graph matching has been used as explained in~\cite{cho2013learning}.
In our paper, we have further extracted 35 subsets of this dataset composed of only the first four annotations ($n=4$, as shown in~\cref{fig:willow}) and four images which were consecutive ($m=4$). Thus, we have $35$ subset graphs. Due to our automated initial graph matching process, while the keypoint locations are always correct, the permutation matrices $\Pm_{ij}$ are contaminated with real noise (\cref{fig:willow_matches}). 

\insertimageStar{1}{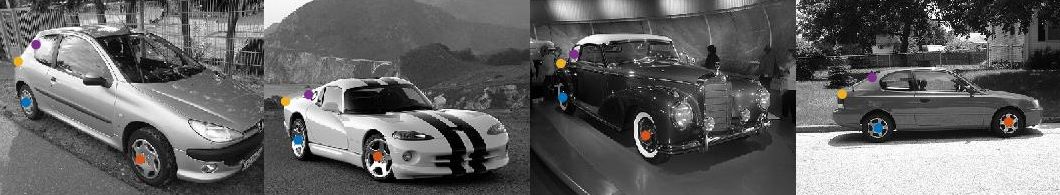}{A random example from the \emph{cars} class of Willow Object Classes dataset~\cite{cho2013learning}.\vspace{-4mm}}{fig:willow}{h!}

\parahead{Extended plots of the evaluations} 
In the paper we have reported the average accuracy over all the images in our modified-willow dataset. However, it is also of interest to see how our algorithm performs on individual subsets. We plot in~\cref{fig:willow_sota} the performance of all the approaches under test, over all the individual subgraphs. It is visible that all the methods perform similarly. While our algorithm has an overall advantage, it can fail on certain examples, despite the quest for the global minimum. This explains the slightly lower accuracy with respect to the exhaustive solution.

\insertimageStar{0.97}{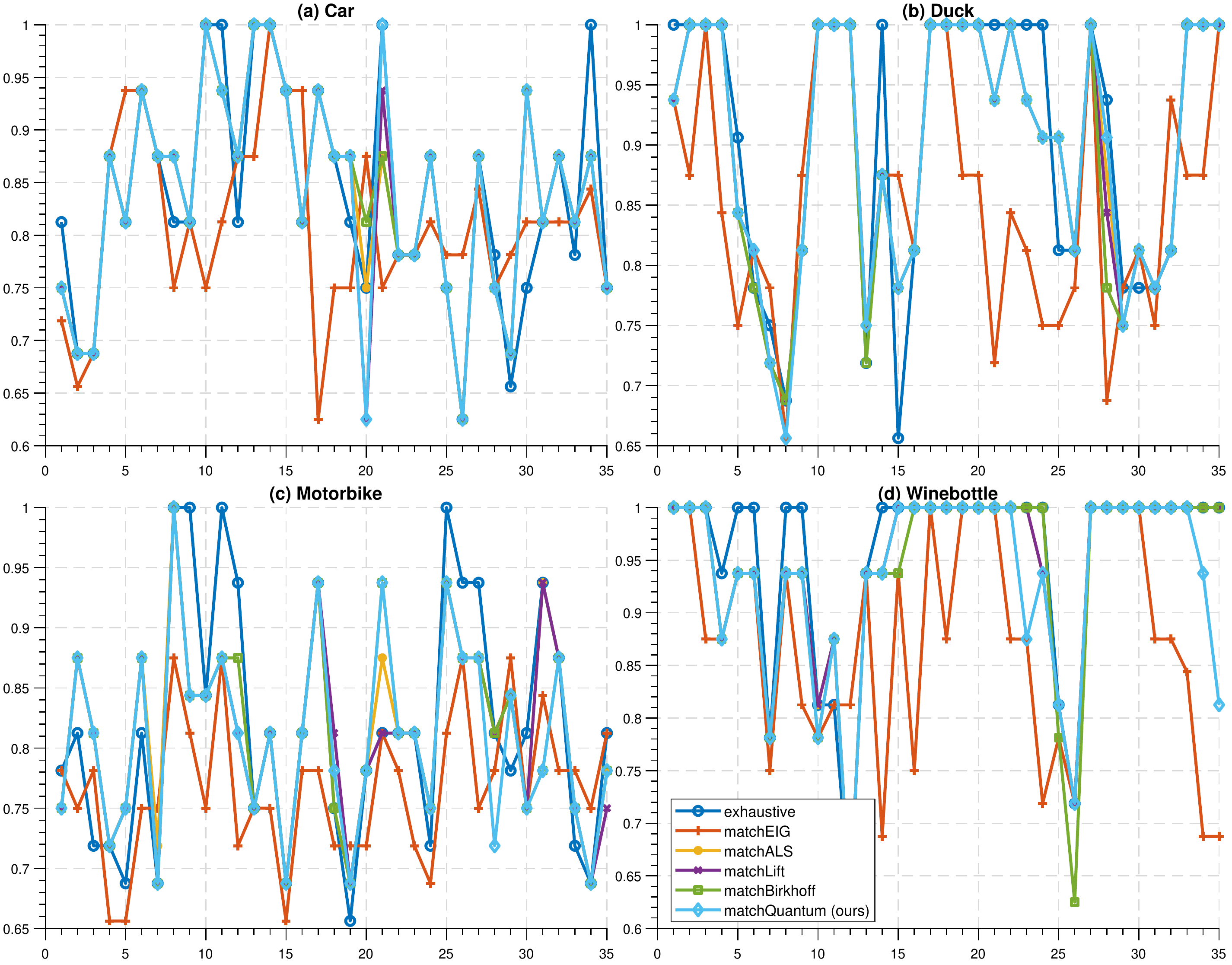}{Detailed evaluations over all the subsets on the Willow Object Classes dataset~\cite{cho2013learning}.\vspace{-3mm}}{fig:willow_sota}{h!}

\insertimageStar{1}{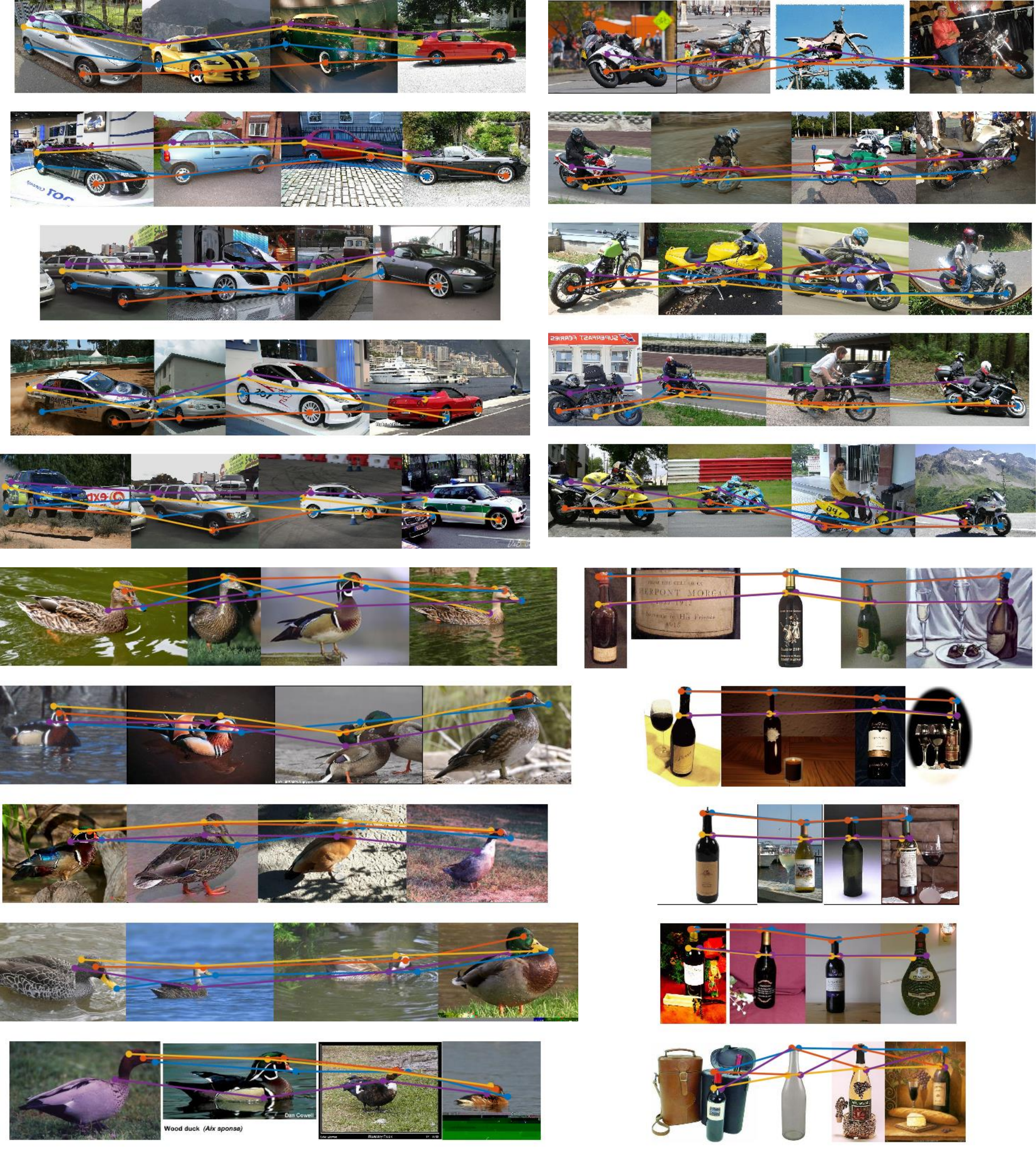}{Initial matches on random examples from Willow Object Classes dataset~\cite{cho2013learning}. }{fig:willow_matches}{t!} 

\parahead{Evaluation metrics} 
In the main paper, we mention that we report the number of bits correctly detected. This so-called \emph{accuracy} measure between the estimated permutations $\X$ and their ground truth $\bar{\X}$ is formally defined as:
\begin{align}
    \epsilon (\X, \bar{\X}) = 1 - \frac{1}{mn^2} \sum\limits_{i=1}^m |\vecm(\X_i \oplus \bar{\X}_i)|_1, 
\end{align}
where $\oplus$ is the \emph{exclusive-or} operand. Intuitively, this is a Hamming similarity derived from the Hamming distance in which an error is made whenever two bits differ.

\section{Detailed Analysis of Several Recent Works on Quantum Computer Vision} 

\begin{table}[h!] 
    \centering 
    \begin{tabular}{|p{3.5cm}|p{6.9cm}|p{1.35cm}|p{1.2cm}|p{2.25cm}|} 
                          $\quad\quad\;\;$\textbf{Algorithm}$\;\;$ & $\quad\quad\quad\quad\quad\quad\quad\quad$\textbf{Problem} & \textbf{Processor} & \textbf{\# qubits} & \textbf{total~QPU time} \\ \hline 
                            QUBO Suppression \cite{LiGhosh2020} & 
                            non-maximum suppression in human tracking &  
                            2X & 
                            1000 & 
                            n/a \\
                            QA \cite{pointCorrespondence, pointCorrespondenceSupp} & 
                            transformation estimation and point set alignment & 2000Q & 
                            2048 & 
                            $60$ sec.$^*$ \\ 
                            QGM \cite{SeelbachBenkner2020} & 
                            graph matching (two graphs, up to four points) &  
                            2000Q & 
                            2048 & 
                            $2-2.5$ min. \\ 
                            \textit{QuantumSync} (ours) &   
                            permutation synchronization (multiple views, multiple points) & 
                            Adv.~$1.1$ & 
                            5436 & 
                            $> 15$ min. \\ 
    \end{tabular} 
    \caption{Overview of several recent quantum computer vision  methods published at computer vision conferences and our \textit{QuantumSync}. 
    Note that the right-most column reports the overall  experimental QPU runtime in the evaluation of the methods. 
    ``$^*$'': QA has been recently tested on D-Wave 2000Q; the results are reported in the supplementary document  \cite{pointCorrespondenceSupp}.} 
    \label{tab:comparison_of_works} 
\end{table} 

In this section, we analyze several recent works on quantum computer vision \cite{pointCorrespondence, LiGhosh2020, SeelbachBenkner2020} and \textit{QuantumSync}. 
We summarize the problems, the D-Wave processors and the experimental QPU time used in the experiments on real quantum hardware in~\cref{tab:comparison_of_works}. 
The QUBO suppression approach \cite{LiGhosh2020} solves an existing QUBO formulation of non-maximum suppression on D-Wave 2X with $1000$ qubits. 
The experimental results show improvements in solving the target combinatorial optimization problem on quantum hardware, and the authors conclude that the age of quantum computing for human tracking has arrived. 

Quantum Approach (QA) to correspondence problems on point sets, inspired by the altered gravitational model for point set alignment \cite{BHRGA2019}, was proposed in \cite{pointCorrespondence}, where the cases with and without known correspondences between the points have been considered. 
In QA, final rigid transformations are approximated as a linear  combination of  basis elements which, in the general case, allows  representing affine transforms. 
Thus, the method can be extended to affine transforms in a straightforward way. 
QA was confirmed on D-Wave 2000Q as reported in the supplementary material  accompanying the paper \cite{pointCorrespondenceSupp}. 
QA can align two point sets at a time. 
Even though the dimensionality of the resulting matrix of couplings and  biases  does not depend on the number of points (the size of this matrix depends on the cardinality of the basis elements  set), the super-linear complexity to prepare the state for large problems on classical hardware can result in large runtimes. 
The method of Benkner~\etal~\cite{SeelbachBenkner2020} for quantum graph matching (QGM) is concurrent to our work. 
While they propose a similar mechanism to impose permutation matrix constraints, there are multiple fundamental differences. 
First, QGM is designed for operation on two graphs, whereas we propose a permutation (map) synchronization algorithm for multiple views/scans with multiple points each. 
In our problem setting, initial (noisy) estimates of the pairwise permutations are given. 
Second, we successfully confirm our method for up to nine views with three points each, or seven views with four points each, whereas in~\cite{SeelbachBenkner2020}, 
only the $3 \times 3$ case has been successfully solved on a real AQC. 
Another distinguishing aspect of our work is that we achieve the probability of over $90\%$ to measure a correct optimal solution for the permutation synchronization problem for the $3 \times 3$ case, whereas~\cite{SeelbachBenkner2020} this probability (for a different problem though) only slightly exceeds the probability of randomly guessing a correct $3 \times 3$ permutation (\textit{i.e.,}  ${\approx}16.7\%$). 
While the experiments in~\cite{SeelbachBenkner2020} are performed on $2000Q$ as opposed to the Advantage system $1.1$ (the latest generation of D-Wave AQC) we use, the difference in the probabilities is unlikely to stem from the differences between the hardware architectures. 
We have tested our algorithm on $2000Q$ as well and even observed slightly higher probabilities for the $3 \times 3$ case, though larger problems either resulted in lower probabilities compared to Advantage system $1.1$ or could not be embedded on the $2000Q$ QPU and solved on it at all. 

\end{document}